\documentclass[superscriptaddress, amsmath,amssymb,aps,prx,floatfix, preprint]{revtex4-1}

\usepackage[english]{babel}
\makeatletter\AtBeginDocument{\let\@elt\relax}\makeatother
\usepackage[utf8]{inputenc}
\usepackage{graphicx}
\usepackage{bm}
\usepackage{xcolor}
\usepackage{hyperref}
\usepackage{wasysym}

\renewcommand{\vec}[1]{\mathbf{#1}}
\newcommand{\matr}[1]{\mathbf{#1}}

\usepackage{tikz,hyperref}
\usepackage{academicons}
\definecolor{orcidlogocol}{HTML}{A6CE39}

\definecolor{lime}{HTML}{A6CE39}
\DeclareRobustCommand{\orcidicon}{%
	\begin{tikzpicture}
	\draw[lime, fill=lime] (0,0) 
	circle [radius=0.16] 
	node[white] {{\fontfamily{qag}\selectfont \tiny ID}};
	\draw[white, fill=white] (-0.0625,0.095) 
	circle [radius=0.007];
	\end{tikzpicture}
	\hspace{-2mm}
}

\foreach \x in {A, ..., Z}{%
	\expandafter\xdef\csname orcid\x\endcsname{\noexpand\href{https://orcid.org/\csname orcidauthor\x\endcsname}{\noexpand\orcidicon}}
}

\begin{document}

\title{Dual communities in spatial networks}

\author{Franz Kaiser\orcidF}
  \affiliation{Forschungszentrum J\"ulich, Institute for Energy and Climate Research (IEK-STE), 52428 J\"ulich, Germany}
  \affiliation{Institute for Theoretical Physics, University of Cologne, 50937 K\"oln, Germany}

\author{Philipp C. Böttcher\orcidP}
  \affiliation{Forschungszentrum J\"ulich, Institute for Energy and Climate Research (IEK-STE), 52428 J\"ulich, Germany}

\author{Henrik Ronellenfitsch\orcidH}
\affiliation{Physics Department, Williams College, 33 Lab Campus Drive, Williamstown, MA 01267, U.S.A.}
\affiliation{Department of Mathematics, Massachusetts Institute of Technology, Cambridge, MA 02139, U.S.A.}
 
\author{Vito Latora\orcidV}
  \affiliation{School of Mathematical Sciences, Queen Mary University of London, London E1 4NS, UK}
    \affiliation{Dipartimento di Fisica ed Astronomia, Universit{\`a} di Catania and INFN, 95123 Catania, Italy}
    \affiliation{Complexity Science Hub Vienna, 1080 Vienna, Austria}

\author{Dirk Witthaut\orcidD}%
  \affiliation{Forschungszentrum J\"ulich, Institute for Energy and Climate Research (IEK-STE), 52428 J\"ulich, Germany}
  \affiliation{Institute for Theoretical Physics, University of Cologne, 50937 K\"oln, Germany}

\date{\today}

\begin{abstract}
Both human-made and natural supply systems, such as power grids and leaf venation networks, are built to operate reliably under changing external conditions.
Many of these spatial networks exhibit community structures.
Here, we show that a relatively strong connectivity between the parts of a network can be used to define a different class of communities: dual communities.
We demonstrate that traditional and dual communities emerge naturally as two different phases of optimized network structures that are shaped by fluctuations and that they suppress failure spreading, which underlines their importance in understanding the shape of real-world supply networks.
\end{abstract}

\maketitle

\section*{Introduction}

Community structures are a fundamental trait of complex networks and have found numerous applications in systems ranging from social networks~\cite{radicchi_defining_2004} to biological networks~\cite{girvan_community_2002,fletcher_network_2013} and critical infrastructures such as power grids~\cite{shekhtman_resilience_2015}. 
Typically, communities are defined by a strong connectivity within the community compared to a relatively weak connectivity between different communities~\cite{radicchi_defining_2004,newman_modularity_2006,newman_communities_2012}.
They may correspond to functional units of the network, for instance in metabolic networks \cite{guimera2005functional}, or actual communities in social networks.

Intuitively, community structures play an important role for the spreading of failures or perturbations in networked systems. 
The low connectivity impedes spreading processes, such that perturbations can be expected to stay within the community, which is both predicted by theory~\cite{Manik2017,may_will_1972} as well as observed in experiments~\cite{gilarranz_effects_2017}.  
Community structures thus provide robustness in complex networks, but other structural patterns may have a comparable effect. In particular, it has been shown that hierarchical structures may provide similar features, for instance in vascular networks of plants \cite{Kati10,gavrilchenko_resilience_2019}. 

In this article we provide a unified view on the role of communities and hierarchies for network robustness based on the concept of graph duality. 
The dual graph is most naturally defined for spatially embedded networks, i.e.~networks that are embedded in the plane without edges crossing each other. 
This class of networks includes a large variety of man-made and biological systems~\cite{viana_simplicity_2013,barthelemy_spatial_2011,Barthelemy2018,Barthelemy2016}. 
The vertices of the dual correspond to the faces of the original, primal graph. 
Two dual vertices (faces) are connected if they share at least one edge. 
Graph duality has been previously used to study fixed points in oscillator networks~\cite{manik_cycle_2017,dorfler_synchronization_2013} and to speed up network algorithms~\cite{ronellenfitsch_dual_2017,17lodf}. 
Here we use this concept to reveal patterns in the network structure that are hidden in the primal graph. 
In particular, we establish dual communities -- communities that are defined within the dual graph -- and highlight their relation to hierarchical network structures. 
Furthermore, graph duality readily explains why both weak and strong connections can make a network robust: Strong connections in the primal translate into weak connections in the dual and vice versa.

The article is organized as follows. 
We first demonstrate how different structural patterns impede spreading processes and thus contribute to the robustness of a network. 
Focusing on flow networks, we formalize the concept of graph duality and establish dual communities. 
Second, we investigate essential properties of dual communities, in particular their link to hierarchical patterns and the geometry of the community boundaries.
Finally, we study the emergence and impacts of communities structures. 
Using optimization models, we study why networks develop a primal or a dual community structure. 
We provide a deeper analysis of the link between communities and network robustness by employing different simulation models. 
Throughout the article we use the terms graph and network interchangeably.
The term `graph' stresses the structural aspects while `network' stresses the functional aspects of the system.

\section*{Results}

\subsection*{Network robustness and community structure}

We first highlight how  network robustness is related to the  presence of communities for selected examples. 
Robustness is essential for critical infrastructures such as electric power grids. 
The high-voltage transmission grid of Scandinavia (Fig.~\ref{fig:connection_robustness}a,b) has an obvious community structure due to geographic reasons. Finland is only weakly connected to the rest of Scandinavia through two high-voltage transmission lines. 
We simulate the impact of a transmission line failure, which is the biggest threat for large-scale blackouts \cite{Pour06}. 
We use the linear power flow or DC approximation \cite{Wood14,Purc05}, which will be described in detail below. 
The flow changes after the failure generally decay with the distance to the failing transmission line, but we also observe a strong impact of the community boundary. 
Flow changes are strongest in the community where the failure occurs, in this case Sweden. 
They are substantially suppressed in the other community, i.e.~Finland, which also reduces the risk of a global cascade of failures. 

Remarkably, a similar suppression of failure spreading is realized by strong connections. We consider the venation network of a leaf, which includes a strong central vein separating the left and right half (Fig.~\ref{fig:connection_robustness}c,d). The flow of water and nutrients is described by a linear flow network~\cite{Kati10}, which is mathematically equivalent to the linear power flow approximation. If a secondary vein fails, we observe a very similar picture as for the power grid: flow changes generally decay and are strongly suppressed in the other half of the leaf. The central vein itself features large flow changes and thus provides a buffer function. We conclude that weak and strong connections can equally suppress the spreading of failures. We will show that both effects are fully equivalent and that they can be understood in terms of network communities, provided we generalize the definition of communities.

\begin{figure*}[tb]
\begin{center}
\includegraphics[width=1.\textwidth]{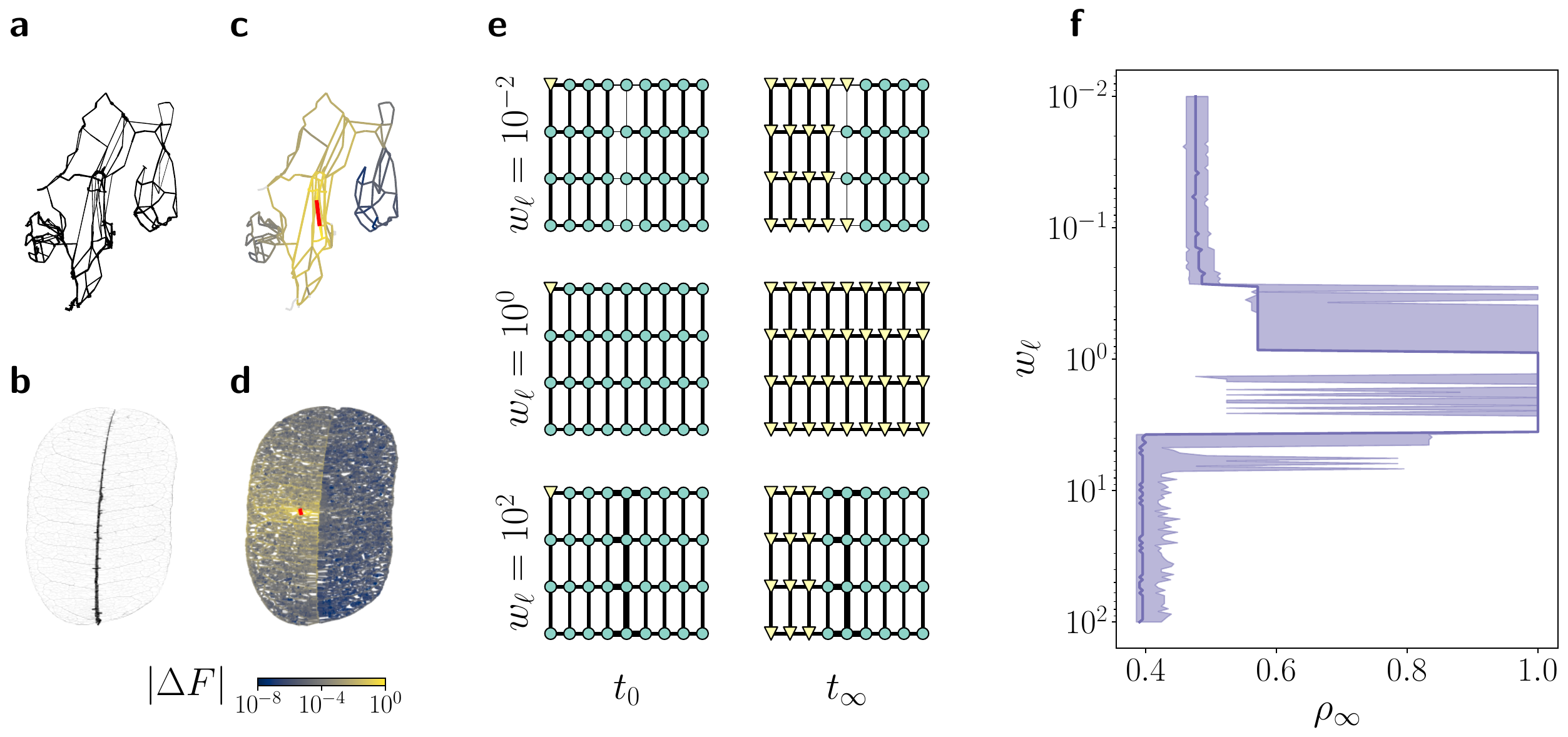}
\end{center}
\caption{\textbf{Different structural patterns separate networks and increase network robustness.} 
\textbf{a} Topology of the Scandinavian power grid, with weak connectivity between different geographic units, in particular Finland.
\textbf{b} Venation network of the leaf Schizolobium amazonicum, with a strong central vein separating the leaf into left and right. The width of the lines in panels (\textbf{a}) and (\textbf{b}) encode the edge weights. 
\textbf{c,d} Flow changes  $|\Delta F|$ after the failure of a single edge (colored in red) for the two networks shown in panels (\textbf{a,b}). The impact of the failure is strongly suppressed in another part of the network, i.e.~in Finland and the right half of the leaf, respectively. This highlights the existence and impact of boundaries that separate the network into communities.
\textbf{e} Simulation of a classic model of global cascades \cite{watts2002simple,nematzadeh_optimal_2014} in a lattice with inhomogenous edge weights. Edges in the middle have weight $w_\ell$ (indicated by thin/thick lines), while all other edges have a weight of  $w_{ij}=1$ (see methods for details). Infected/faulty nodes are shown as yellow triangles, healthy/operational nodes as green circles. 
A global cascade occurs for a homogeneous lattice ($w_\ell=1$), while both weak ($w_\ell=10^{-2}$) and strong connections ($w_\ell=10^2$) stop the cascade from propagating to the right part of the network.
\textbf{f} Final fraction $\rho_\infty$ of nodes that become infected/faulty during the cascade as a function of the weight parameter $w$. 
The line represents the median and the shaded region the 25\% to 75\% quantile for $1000$ random initial conditions. For homogeneous lattices ($w=1$), the cascade reaches all nodes ($\rho_\infty=1$). For both weak ($w \ll 1$) and strong ($w \gg 1$) conditions, the cascade stops at the boundary such that $\rho_\infty \approx 0.5$
}
\label{fig:connection_robustness}
\end{figure*}

Before we move to a more detailed analysis, we demonstrate the generality of the observed phenomena. 
We consider a classic model of network cascades \cite{watts2002simple,nematzadeh_optimal_2014}. 
Nodes are either healthy/operational (state $0$) or infected/faulty (state $1$). 
A node $i$ gets infected or faulty if the weighted average over all neighbors' states exceeds a certain threshold $\phi_i$. 
Starting from a small amount of nodes in state $1$, a cascade may emerge depending on the values $\phi_i$ and the structure of the network. 
As before, we consider networks which can be separated into two parts by either weak or strong connections (Fig.~\ref{fig:connection_robustness}e,f). 
More precisely, we consider a lattice where the edges in the middle region have a tunable weight $w$, while all other edges have weight one (see Methods for details). 
We observe that a homogeneous network ($w=1$) always leads to a global cascade, where all nodes are infected or faulty in the final state. 
A boundary, either by weak ($w\ll1$) or strong ($w\gg1$) connections, effectively stops the cascade. Only those nodes are infected, which are located in the same half as the initially infected ones.  

\subsection*{Flow networks and dual communities}

We have shown that strong connections can divide a network and enhance its robustness similarly as weak connections do. Even more, we can establish a full mathematical equivalence of weak and strong connections in the case of flow networks. 
This equivalence leads to a generalization of the definition of community structures in complex networks.   

Linear flow networks arise in a variety of applications, including electric circuits~\cite{bollobas1998,Dorfler2018}, power grids~\cite{Wood14,Purc05},  hydraulic networks~\cite{Hwan96,diaz2016}, and vascular networks of plants~\cite{Kati10}. In these networks, the flow from node $i$ to node $j$  is given by $F_{i \rightarrow j} = w_{ij} \cdot (\theta_i - \theta_j)$, where $w_{ij}$ is the connectivity or conductivity of the edge $(i,j)$. The nodal variable $\theta_i$ describes the local voltage or potential in an electric circuit, the voltage phase angle in a power grid, or the pressure in a hydraulic or vascular network. The flows have to satisfy the continuity equation (or Kirchhoff's current law, KCL) at every node $i$ of the network, $\sum_j F_{i \rightarrow j} = P_i$, where $P_i$ is the inflow to the network.

These equations can be recast in a compact matrix notation. Let $N$ denote the number of nodes and $M$ the number of edges the network, which we assume to be connected. We fix an orientation for each edge to keep track of the direction of flows and define the edge-node incidence matrix $\matr{I} \in \mathbb{R}^{M \times N}$ as
\begin{equation}
  I_{\ell n} = \left\{
   \begin{array}{r l}
     +1 & \; \text{if edge $\ell$ starts at node $n$},  \\
     - 1 & \; \text{if edge $\ell$ ends at node $n$},  \\
      0     & \; \text{otherwise}.
  \end{array}
  \right.
  \label{eq:incidence_matrix} 
\end{equation}
The edge weights are summarized in a diagonal matrix $\matr{W} = {\rm diag}(w_1,\ldots,w_M)$ while all other quantities are summarized in vectors $\vec \theta = (\theta_1, \ldots, \theta_N)^\top$, $\vec P = (P_1, \ldots, P_N)^\top$, $\vec F = (F_1, \ldots, F_M)^\top$. 
Note, the ordering of the edges in the edge-node incidence matrix $\matr{I}$ and the weights in the diagonal matrix $\matr{W}$ have to be consistent such that the weight of edge $k$ connecting nodes $i$ and $j$ is given by $w_k = w_{ij}$.
Then the relation of flows and potentials is given by $\vec F = \matr{W} \matr{I} \vec \theta$ and Kirchhoff's current law reads
\begin{equation}
    \vec P = \matr{I}^\top \vec F =
    \underbrace{\matr{I}^\top \matr{W} \matr{I} }_{=: \matr{L}} \, \vec \theta \, .
    \label{eq:Poisson-primal}
\end{equation}
Equation \eqref{eq:Poisson-primal} is a discrete Poisson equation that determines the potential $\vec \theta$ up to an irrelevant additive constant. The matrix $\matr{L} \in \mathbb{R}^{N \times N}$ is nothing but the well known graph Laplacian with components
\begin{equation}
    L_{ij}   =\left\{\begin{array}{l l }
      -w_{ij} & \; \mbox{if $i$ is connected to $j$},  \\
      \sum_{k} w_{ik} & \; \mbox{if $i=j$},  \\
      0     & \; \mbox{otherwise}.
  \end{array} \right. \label{eq:Laplacian}
\end{equation}
The Laplacian is a central object in spectral graph bisection \cite{fortunato_community_2010}, a classic method of community detection, which will be further elaborated below.

The above description focuses on the nodes of the network, with the nodal potentials $\vec \theta$ being the central quantity of interest. An equivalent description exists that focuses on the edges of the network and the flows $\vec F$.
The starting point is the KCL $\matr{I}^\top \vec F = \vec P$. This linear set of equations is underdetermined in terms of $\vec F$, such that the general solution can be written as the sum of a particular solution and an arbitrary solution of the associated homogeneous equation, namely
\begin{equation}
    \vec F = \vec F_{\rm part} + \vec F_{\rm hom} \, .
    \label{eq:KCL-general-solution}
\end{equation}
The vector $\vec F_{\rm hom}$ describes a flow without sources or sinks, that is, a collection of cycle flows. 
The cycle flows form a vector space (the cycle space) such that we can expand each cycle flow into a suitable basis. 
A distinguished basis exists for plane graphs, i.e.~graphs embedded in the plane, which can be constructed in the following way.
A face of a plane graph is a region that is bounded by edges, but contains no edges in the interior. 
The boundary edges of each face then provide one basis vector of the cycle space. 
Further details are given in the supplementary information.

To keep track of the basis, we introduce the cycle edge incidence matrix $\matr{C} \in \mathbb{R}^{M \times (M-N+1)}$ with components
\begin{equation}
  C_{\ell c}= \left\{
   \begin{array}{r l }
     +1 & \; \mbox{if edge $\ell$ is part of face $c$},  \\
     -1 & \; \mbox{if reversed edge $\ell$ is part of face $c$},  \\
      0     & \; \mbox{otherwise}.
   \end{array} \right. 
\end{equation}
Then we can write the general solution of the KCL as 
\begin{equation}
    \vec F = \vec F_{\rm part} + \matr{C} \, \vec f
    \label{eq:F-Cf}
\end{equation}
with an arbitrary cycle flow vector $\vec f$. 
The actual values of the cycle flows are then determined by Kirchhoff's voltage law (KVL), which states that the potential differences around any closed cycle sum up to zero. 
In fact it is sufficient to enforce this for the $M-N+1$ basis cycles. We can thus formulate the KVL in terms of the flow vector $\vec F$ as
\begin{equation}
    \matr{C}^\top \matr{W}^{-1} \vec F = \vec 0.
\end{equation}
Crucially, this equation includes the matrix $\matr{W}^{-1}$ which translates flows into potential differences. 
Substituting Eq.~\eqref{eq:F-Cf} then yields 
\begin{equation}
    \vec Q = \underbrace{\matr{C}^\top \matr{W}^{-1} \matr{C}}_{=\matr{L}^*} \, \vec f   \, .
    \label{eq:Poisson-dual}
\end{equation}
Notably, this equation has the same mathematical structure as Eq.~\eqref{eq:Poisson-primal}: It is a discrete Poisson equation with a Laplacian matrix $\matr{L}^*$ and a source term $\vec Q = -\matr{C}^\top \matr{W}^{-1} \vec F_{\rm part}$. However, the Laplacian $\matr{L}^*$ is not defined on the original primal graph, but on the dual graph. The vertices of this dual graph are given by the faces of the primal graph, while two nodes of the dual graph are connected by a dual edge if the corresponding faces share an edge.

Comparing the Laplacian of the primal graph $\matr{L}=\matr{I}^\top \matr{W} \matr{I}$ to that of  the dual graph $\matr{L}^*=\matr{C}^\top \matr{W}^{-1} \matr{C}$, we see another essential aspect of graph duality: The weights of the dual edges are inverse to the weights of the primal edges. More precisely, we find the dual weights
\begin{equation}
    w^*_{c,d} = \sum_{\ell \in c,d} \frac{1}{w_\ell}
    \label{eq:dual-weight}
\end{equation}
of the edge that connects the nodes $c$ and $d$ in the dual graph corresponding to faces $c$ and $d$ that share the edge $\ell$ in the primal graph. 
This relation shows most clearly why weak and strong connections can both affect the robustness and the community structure of a network. 
Strong connections in the primal correspond to weak connections in the dual and vice versa. Similarly, a strong central vein in the primal corresponds to weak connections in the dual and thus to a pronounced community structure.

We have now introduced all mathematical tools to identify dual communities in planar complex networks. Starting from the primal network, we identify all faces and define the dual graph with weights given by Eq.\eqref{eq:dual-weight}. Then dual communities can be extracted by means of any standard community detection algorithm.

\begin{figure*}[tb]
\begin{center}
\includegraphics[width=1.\textwidth]{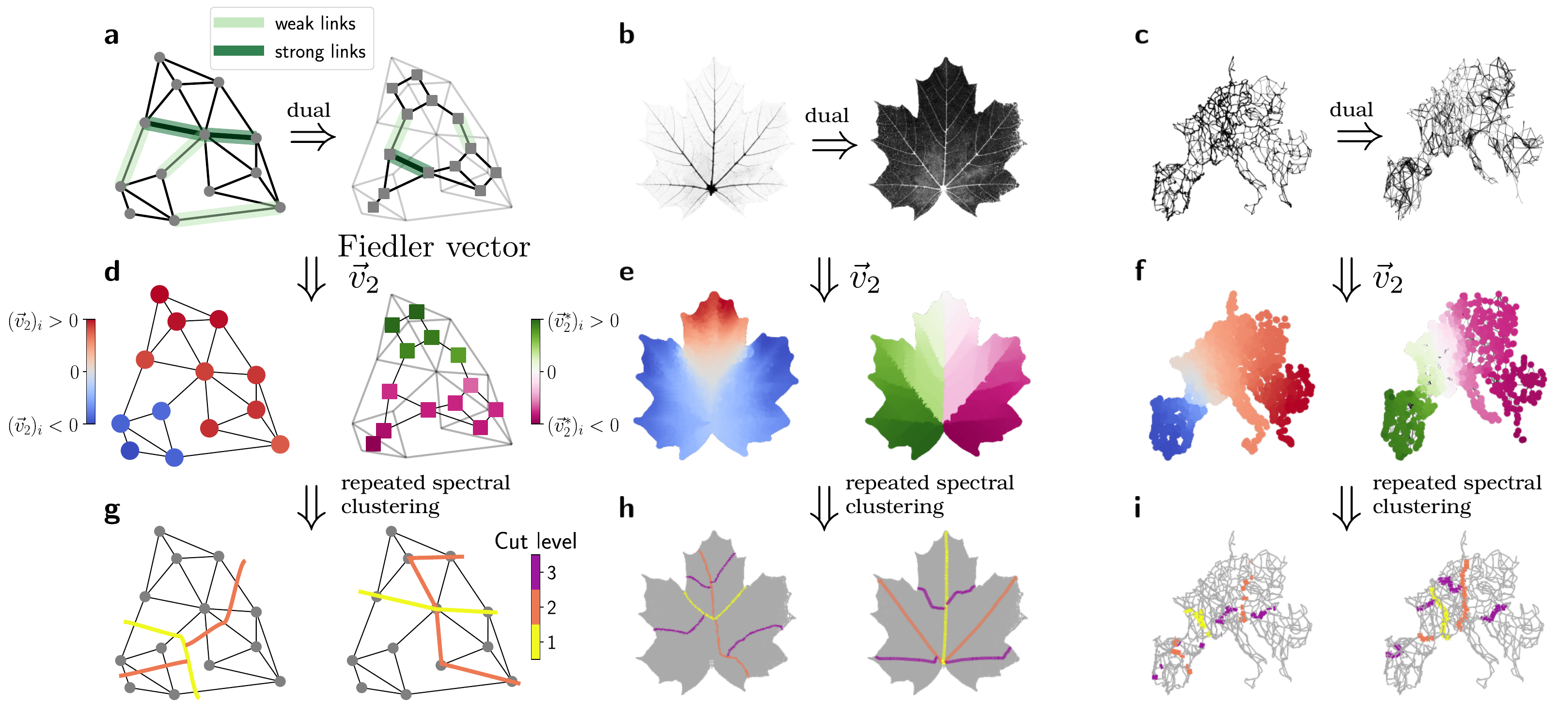}
\end{center}
\caption{\textbf{Primal and dual communities and hierarchies in spatial networks.} 
\textbf{a} A plane graph with edges characterized by either large (dark green), small (light green) or intermediate edge weights, and its associated dual graph. 
The dual graph is constructed by transforming each face of the graph into a node of the dual, and adding dual edges whenever two faces share an edge. Notice that an edge with large weight shared by two faces will imply a weak link between the two corresponding nodes in the dual graph (see Eq.~\eqref{eq:dual-weight}). 
\textbf{d} Spectral clustering by means of the Fiedler vector $\vec{v}_2$ reveals the community structure in both the graph (left) and its dual (right). 
\textbf{g} Based on repeated spectral clustering, the graphs are further decomposed into a hierarchy of smaller sub-units which is different in the graph (left) and its (dual). 
\textbf{b,e,h} If we perform the same analysis on the venation network of a leaf of \textit{Acer platanoides}, a decomposition of the original graph does not provide useful information (\textbf{e,h},left). 
A decomposition of the dual graph, however, reveals the hierarchical organization and the functional units of the venation network (\textbf{e,h}, right). 
\textbf{c,f,i} Applying the same procedure to the Central European power grid,  (\textbf{i},left) and its dual (\textbf{i},right) show that primal and dual hierarchies provide different, but equally useful information about the network structure.
We may conclude that this network represents an intermediate case between primal and dual community structure (see main text for details).
}
\label{fig:power-and-leaf}
\end{figure*}

In the following, we focus on spectral graph bisection because of its direct link to the graph Laplacian -- which is the central object in the above analysis. This method relies on the fact that the community structure is encoded in the second smallest eigenvalue of the graph Laplacian $\lambda_2\geq 0$, known as the algebraic connectivity or Fiedler value, which vanishes if the graph consists of two disconnected components and increases with increasing connectivity between the communities.
The graph nodes are then assigned to one of two communities based on the corresponding eigenvector $\vec{v}_2$, the Fiedler vector: two vertices $j$ and $i$ are in the same community if they share the same sign of the Fiedler vector~\cite{newman_modularity_2006}.
\begin{equation}
  \operatorname{sign}((\vec{v}_2)_i-h)=\operatorname{sign}((\vec{v}_2)_j-h),
\end{equation}
where $h\in \mathbb{R}$ is a threshold parameter. Here, we choose $h=0$. 

This method can be straightforwardly applied to the dual graph, replacing the primal Laplacian $\matr{L}$ by its dual counterpart $\matr{L}^*$. The algebraic connectivity of the dual is measured by the second eigenvalue $\lambda_2^*$ and the associated eigenvector is used to identify the dual communities. 
We find that dual communities appear naturally in real-world networks such as the venation networks of leaves (Fig.~\ref{fig:power-and-leaf}b,e). In the following, we discuss essential properties of dual communities, in particular their relation to hierarchical structures, and provide a thorough analysis of the dual algebraic connectivity $\lambda_2^*$. We stress that other community detection methods can be applied to the dual graph equally well and yield comparable results (cf.~Supplementary Figure~S6). 

\subsection*{Dual communities reveal hierarchical organization of supply networks} \label{sec:hierarchies} 

The spectral clustering method for community detection can be applied to both the primal and the dual graph, revealing different structural information about the network (Figure~\ref{fig:power-and-leaf}). 
Furthermore, we can use this approach to extract a network's hierarchical organization as follows. 
Starting from the initial network, we compute the Fiedler vector, identify the communities and then split the network into two parts at the resulting boundary by removing all edges between the communities. Then we iterate the procedure starting from the subgraphs obtained in the previous step. Repeated application of this procedure reveals different boundaries and thus different hierarchies in the primal graph and in its dual (Figure~\ref{fig:power-and-leaf}g, see Methods for details).

Leaf venation networks are archetypal examples of hierarchically organized networks, with a thick primary vein in the middle and medium secondary veins, that supply thin subordinated veins (Figure~\ref{fig:power-and-leaf}b). 
The thick veins separate the network into distinct parts -- for instance the left and right half separated by the primary vein. 
This characteristic organization is clearly revealed by dual community detection. 
Spectral graph bisection identifies the primary vein that separates the left and the right half of the leaf. 
Repeating the bisection then shows that this organizational pattern repeats in a hierarchical order: dual communities are split by secondary veins in a repeated manner (Fig.~\ref{fig:power-and-leaf}h). 
Remarkably, an analogue decomposition in the original primal graph does not provide any useful information on the network organization.

We conclude that leaf venation networks clearly display a dual community structure, where the boundary of the dual communities coincide with the primary and secondary veins. Hence, dual community detection allows to identify hierarchical organization patterns in complex networks. We will provide a more formal treatment of the relation between strong veins and dual communities below.

As a second example, we now turn to another type of spatially embedded supply networks: power grids. Figure~\ref{fig:power-and-leaf}c shows the European power transmission grid and its dual graph. Again, a hierarchical decomposition reveals different levels of hierarchies in the grid that correspond to its functional components. These components may also be interpreted geographically: the mountain ranges such as the Pyrenees or the Alps as well as the former Iron Curtain are clearly visible in the decomposition of the primal graph. Remarkably, both primal and dual decompositions provide useful structural information here. In particular, there is a dual community boundary at cut level three that spans Hungary and the border region between Slovenia and Croatia and closely corresponds to a weak spot in the European power system, where the grid was split into two mutually asynchronous fragments on January, 8th, 2021 \cite{entsoe2021a}. Another split occurred on the 24th of July between Spain and France -- where both the primal and the dual decomposition detect a community boundary \cite{entsoe2021b}.

Although mathematically similar~\cite{gavrilchenko_resilience_2019,Kaiser2019,kaiser_optimal_2020}, the two types of networks we studied display different structural hierarchies and communities. Whereas leaf venation networks are evolutionarily optimized, the structure of power grids depends strongly on historical aspects and their ongoing transition to include a higher share of renewable energy sources. This transition aspect also manifests in their community structure, as we will see further below.

\subsection*{Connectivity and the geometry of community boundaries}

\begin{figure}[tb]
\begin{center}
\includegraphics[width=0.8\columnwidth]{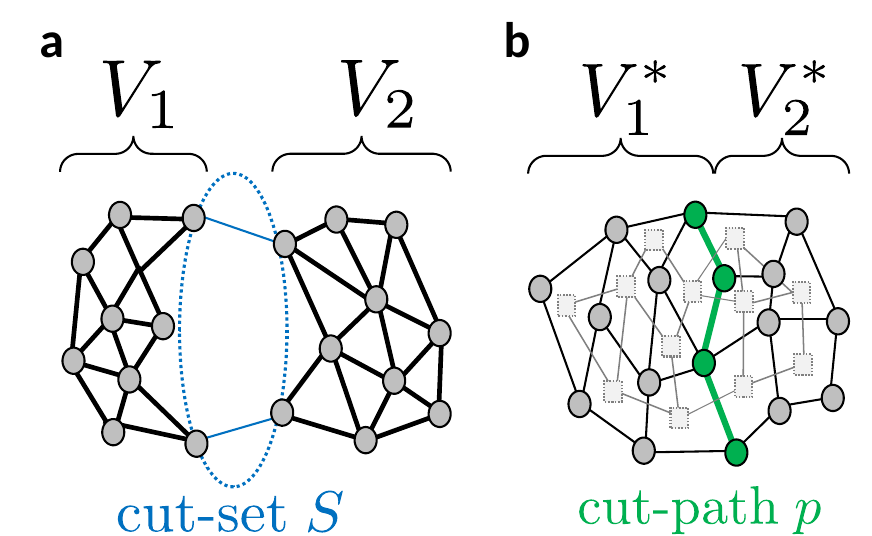}
\end{center}
\caption{\textbf{Geometry of primal and dual community boundaries.}
\textbf{a} Decomposition of a graph into two primal communities. Some edges belong to neither of the two communities, but provide a weak connection between the communities. These edges constitute the cut-set $S$.
\textbf{b} A graph (black solid lines) and its dual (grey dashed lines), which is decomposed into two communities $V_{1,2}^{*}$. 
Some primal edges, colored in green, belong to faces from both communities. These primal edges, together with their terminal nodes, constitute the cut-path $p$.
}
\label{fig:geometry}
\end{figure}

The algebraic connectivity $\lambda_2$ of a graph is closely related to its topological connectivity -- the amount of connectivity between the two communities \cite{Fiedler1973}. For a weighted graph, one can derive the upper bound (see Supplementary Note~2)
\begin{equation}
   \lambda_2 \le \mu_2 = \frac{N_1 + N_2}{N_1 N_2} \sum_{\ell \in S} w_\ell \, .
   \label{eq:lambda-bound-primal}
\end{equation}
where $N_{1}$ and $N_2$ respectively count the number of nodes in the two communities.
The set $S$ contains all edges which are not within one of the two communities but in-between, providing a weak connection of the communities (Fig.~\ref{fig:geometry}a). This set is referred to as a cut-set: If all edges in $S$ are removed, the graph is cut into the two communities. Notably, the bound becomes becomes exact approximation in the limit of vanishing connectivity ($\mu_2 \rightarrow 0$) as shown in the supplementary information.

We derive an analogous bound for dual communities, transferring geometric concepts from the primal to the dual. In particular, we derive an analog to the cut-set $S$, which contains all edges, which are elements of neither of the two components.
Consider a decomposition of the dual graph $G^* = (V^*,E^*)$, where the dual vertex set $V^*$ is separated into two components $V_1^*$ and $V_2^*$. Two faces $c \in V_1^*$ and $d \in V_2^*$ are connected in the dual, if they share at least one edge in the primal graph. Hence, we will find a set of primal edges which belong to both of the two components (Fig.~\ref{fig:geometry}b). These primal edges, together with their terminal vertices, constitutes a path $p$ in the primal graph. In the following, we will refer to $p$ as a cut-path as its removal disconnects the graph. 
The edges along the cut-path essentially determine the community structure of the dual graph and its algebraic connectivity. Given a cut-path $p$, we find the bound
\begin{equation}
   \lambda^*_2 \le \mu_2^* = \frac{N^*_1 + N^*_2}{N^*_1 N^*_2} \sum_{\ell \in p} \frac{1}{w_\ell} \, ,
      \label{eq:lambda-bound-dual}
\end{equation}
where $N^*_{1,2} = |V^*_{1,2}|$ counts the number of nodes in the dual communities. Notably, the expression $\mu_2^*$ does not only provide an upper bound for the algebraic connectivity $\lambda_2^*$, but an approximation that becomes exact in the limit of vanishing dual connectivity. We prove these statements rigorously in the supplementary information.  

\begin{figure*}[tb]
\begin{center}
  \includegraphics[width=1.0\textwidth]{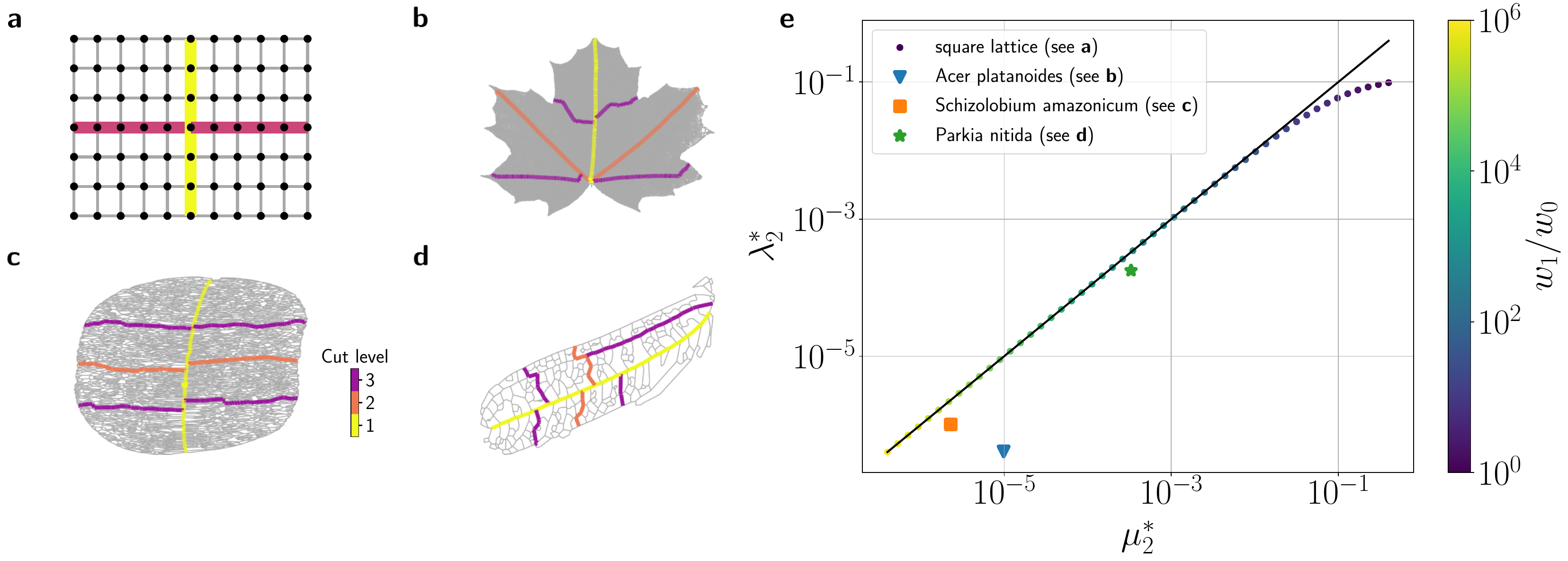}
\end{center}
\caption{\textbf{Algebraic and Topological Connectivity in the dual graph of synthetic and real-world networks.}
\textbf{a} A square lattice with a tunable dual community structure. The edge weights along the central vein $w_1$ are stronger than the weights $w_0$ of the remaining edges. The ratio $w_1/w_0$ can be tuned to test the relation of algebraic and topological connectivity given by  Eq.~\eqref{eq:lambda-bound-dual}. 
\textbf{b-d} Hierarchical organization in leaf venation networks of three different species revealed by repeated graph bisection of the dual graph.
The dual community boundary (yellow) constitutes a cut-path $p$.
\textbf{e} Comparison of the dual algebraic connectivity $\lambda^*_2$ and the dual topological connectivity 
$\mu_2^*$ defined in Eq.~\eqref{eq:lambda-bound-dual}.
Circles refer to the case of the square lattice with the color of the circles indicating the fraction $w_1 / w_0$.
The topological connectivity $\mu_2^*$ provides a rigorous upper bound for $\lambda_2^*$, but also a good approximation for large parameter regions. 
This correspondence shows how dual communities are decomposed by a strong connectivity along the boundary.
\label{fig:connectivity}
}
\end{figure*}
 
The relation of cut-paths and dual communities is further investigated in Fig.~\ref{fig:connectivity} for both synthetic networks and leaf venation networks.
We first consider a square lattice with a tunable dual community structure: The edges in the central vein have a higher weight $w_1$ that the remaining edges $w_0$. We find that the dual algebraic and topological connectivity $\lambda_2^*$ and $\mu_2^*$ become virtually indistinguishable for $w_1/w_0 \apprge 10^{2}$.
In venation networks, the boundaries between the dual communities, i.e.~the cut-paths, correspond to the primal and secondary veins as described above. A good agreement between $\mu_2^*$ and $\lambda^*_2$ is found especially for the two smaller venation networks from the Parkia and Schizolobium family. 
This result further emphasizes the intimate relation of dual communities and hierarchical organization in complex networks.

\subsection*{Why do primal and dual communities emerge?} 

\begin{figure*}[ht!]
\begin{center}
  \includegraphics[width=.85\textwidth]{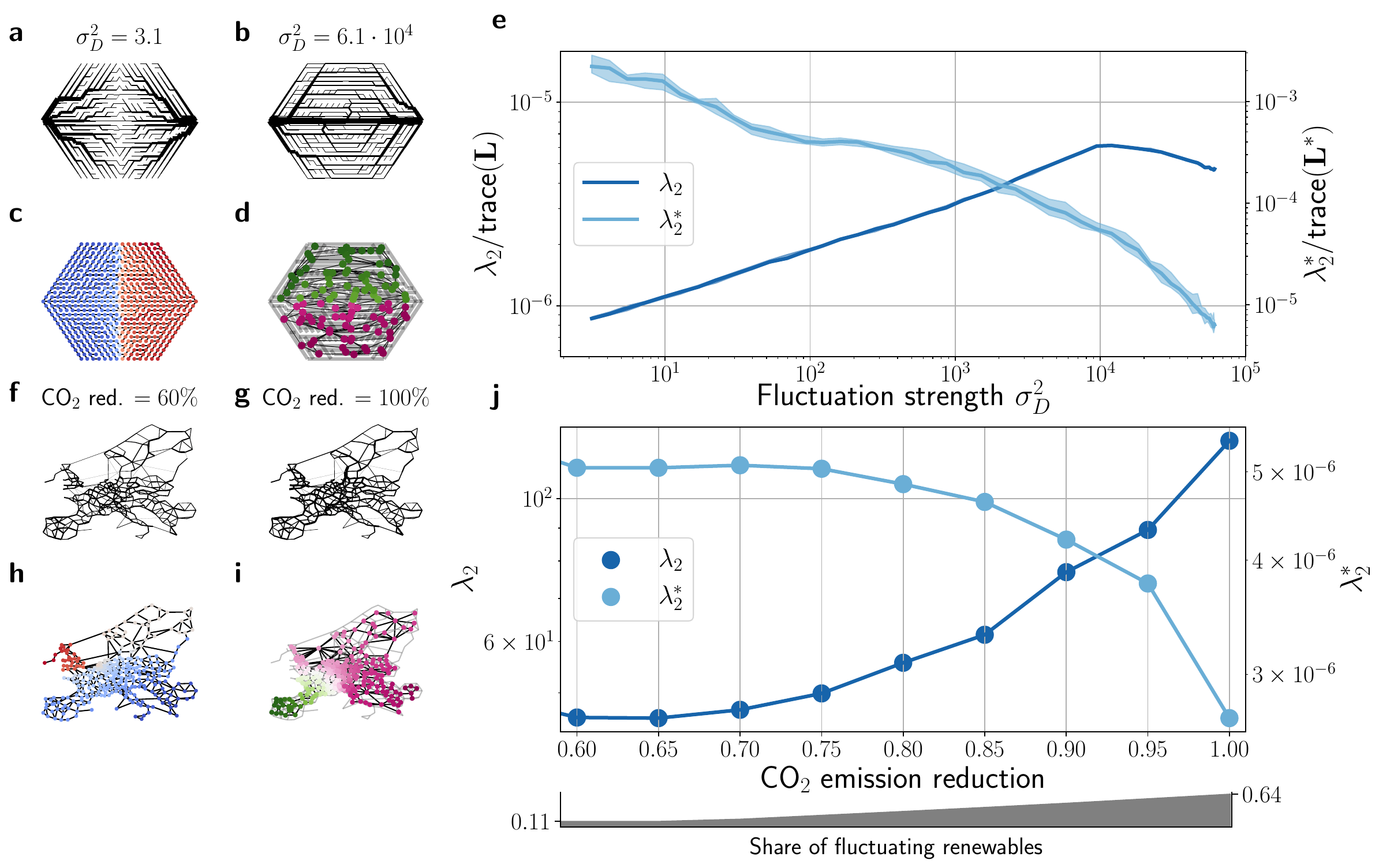}
\end{center}
\caption{\textbf{Primal and dual communities emerge naturally in optimal supply networks.}
\textbf{a-b} We consider a triangular lattice with two fluctuating sources located at the leftmost and the rightmost node and Gaussian sinks attached to all other nodes. The strength of the two sources fluctuates following a Dirichlet distribution with variance $\sigma_D^2$ (see methods). For each value of $\sigma_D^2$, we find the network structure and the edge weights that minimize the total dissipated energy $D$ assuming limited resources. 
\textbf{c-d} The optimal network structure shows a transition from primal to dual communities measured by the Fiedler vectors (colour code) of  the primal graph (panel \textbf{c}) or dual graph (panel \textbf{d}).
\textbf{e} The scaling of the corresponding primal ($\lambda_2$, dark blue) and dual ($\lambda_2^*$, light blue) Fiedler value confirms the transition. The shaded regions indicate the 25\% to 75\% quantiles for different runs.
\textbf{f-g} A transition from primal to dual community structure is also observed in optimization models of the European power grid when generation shifts to fluctuating renewables. The figure shows the  optimal network structures and transmission line capacities in a cost optimal system for different carbon-dioxide (CO$_2$) emission reduction levels. 
\textbf{h-i} Primal and dual communities are identified  by the Fiedler vectors (colour code) of (\textbf{h}) the primal or (\textbf{i}) dual graph.
\textbf{j} The transition is confirmed by the scaling of primal and dual Fiedler values for increasing emission reduction which corresponds to an increasing share of fluctuating renewables (lower axis). 
Simulations were performed with the high-resolution European energy system model 'PyPSA-EUR'~\cite{horsch_2018}, that optimizes the investments and operations of generation and transmission infrastructures for minimum system costs.
\label{fig:Optimal_supply_networks} 
}
\end{figure*}

Understanding how the structure of optimal supply networks emerges is an important aspect of complex networks research~\cite{Kati10,Corson2010,bohn_structure_2007,dodds2010optimal}.
In cases where a single source supplies the entire network, it is well established that fluctuations in the supply can cause a transition from a tree-like topology to a structure with loops~\cite{Kati10,Corson2010,kaiser_optimal_2020}. We extend this result by studying how the increase in fluctuations influences the optimal network structure in supply networks with multiple strongly fluctuating sources and weakly fluctuating sinks. This design is highly relevant for many real-world applications, e.g. when considering a power grid that is based on decentralized renewable energy sources that fluctuate more than conventional carriers.

To interpolate between strongly fluctuating sources and weakly fluctuating ones, we first use a similar model as in Ref.~\cite{Corson2010}. We consider a linear flow network consisting of a triangular lattice with $N$ nodes of which $N_s$ are sources and $N-N_s$ are sinks whose outflows are fluctuating iid Gaussian random variables. Additionally, we add fluctuations only to the sources of the networks that can be tuned by the additional variance $\sigma_D^2$ (see Methods). 
We then compute the optimal structure and edge weights of the network that minimize the total dissipated energy $D = \sum_\ell \langle F_\ell^2 \rangle/w_\ell$ averaged over the fluctuating in- and outflows. Resources for building the network are assumed to be limited, which translates into the constraint $\sum_e w_e^\gamma \le 1$. The cost parameter $\gamma$ quantifies how expensive the increase of an edge weight is and was set to $\gamma=0.9$ for the examples presented in this manuscript (see Supplementary Note~4 for more information).
Results for $N_s=2$ sources are shown in Fig.~\ref{fig:Optimal_supply_networks}, and further results for $N_s=3$ are provided in the supplementary information.

We find that the optimal network structure changes strongly as the fluctuations increase. For weak fluctuations, $\sigma_D^2\approx 1$, each of the $N_s$ sources supplies the surrounding area of the network. Only weak connections are established between the areas to cope with the small residual imbalances. Hence, the optimal networks show a pronounced primal community structure see Fig.~\ref{fig:Optimal_supply_networks}a).

For strong fluctuations, $\sigma_D^2\gg 1$, a local area supply is no longer possible and long-distance connectivity is required. Remarkably, this connectivity is established in one central vein that links the two fluctuating sources see Fig.~\ref{fig:Optimal_supply_networks}b). 
As a consequence, the optimal networks show a pronounced dual community structure similar to leaf venation networks. 
We can capture the transition from a primal to a dual community structure in terms of the primal and dual Fiedler values (Fig.~\ref{fig:Optimal_supply_networks}e). 
Increasing $\sigma_D^2$, we observe a smooth crossover from primal communities with $\lambda_2 \rightarrow 0$ to dual communities with $\lambda^*_2 \rightarrow 0$. We note that a similar picture is found if the Fiedler values $\lambda_2$ and $\lambda_2^*$ are replaced by another measure such as the modularity (see Supplementary Figure~S1). 
We conclude that optimal supply networks typically have a community structure -- whether it is primal or dual depends on the degree of fluctuations.

Strikingly, an analogous transition is observed for actual power transmission grids when optimizing the network structure for different levels of fluctuating renewable energy sources. We consider the European power transmission grid and optimize its network structure for different carbon dioxide (CO$_2$) emission reduction targets compared to the year 1990 ranging from $60\%$ to $100\%$ reduction using the open energy system model 'PyPSA-Eur'~\cite{horsch_2018} (see Methods for details). In Supplementary Figures~S3 and S4 we illustrate how the generation mix in the optimized power system changes for different emission scenarios from conventionally-dominated grids to highly-renewable grids. 

We find that the decarbonization of power generation drives a transition from primal to dual communities in the grid. A reduction in generation-based CO$_2$ emissions corresponds to an increased share of power being produced by fluctuating renewable energy sources. With increasing penetration of fluctuating renewables, we observe a decrease in the dual Fiedler value $\lambda_2^*$ and an increase in the primal Fiedler value $\lambda_2$, which indicates a transition from primal to dual communities in the optimized networks (Fig.~\ref{fig:Optimal_supply_networks}j). 
Hence, the primal-dual transition emerges both in fundamental models and in realistic high-resolutions simulations of spatial networks.

\subsection*{How do primal and dual communities determine network robustness?}

\begin{figure*}[tb]
\begin{center}
  \includegraphics[width=1.0\textwidth]{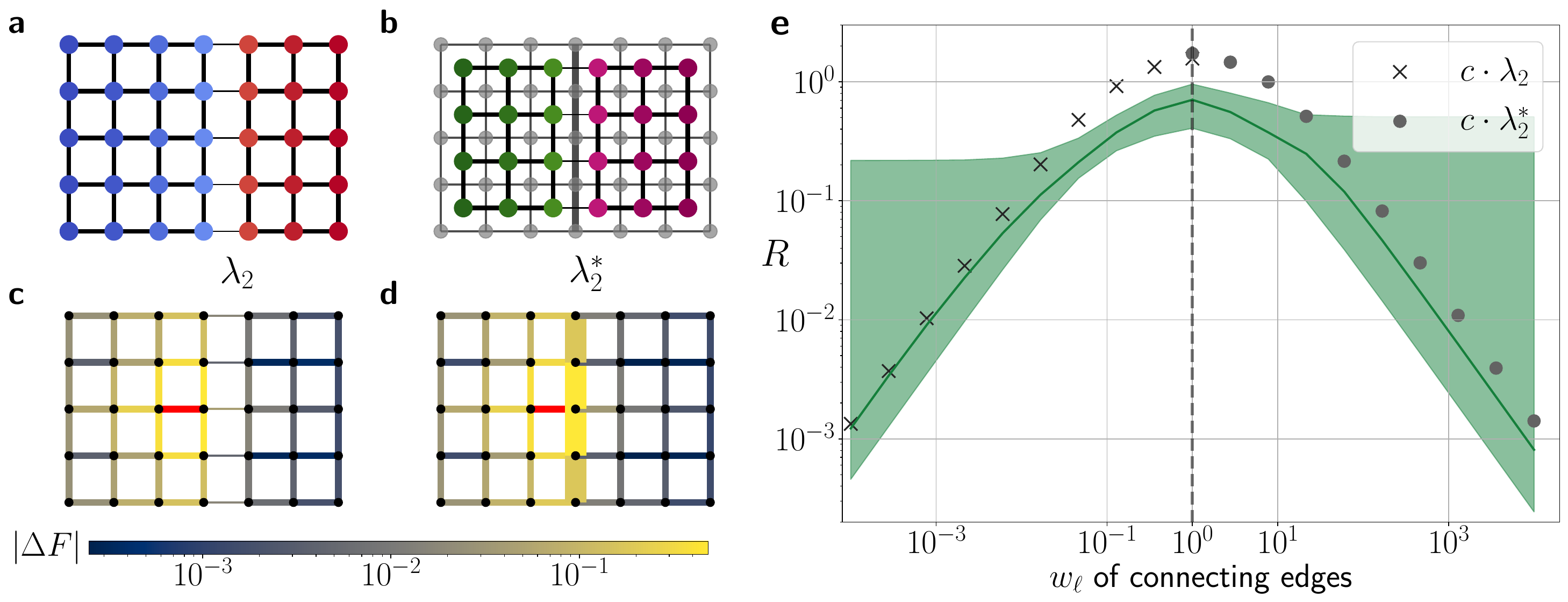}
\end{center}
\caption{\textbf{Primal and dual communities inhibit failure spreading.} 
\textbf{a,b}
A square grid is divided into (\textbf{a}) two primal communities by weakening the central horizontal edges or (\textbf{b}) into two dual communities by strengthening the central vertical edges. The Fiedler vector (colour code) reveals the community structure. 
\textbf{c,d}
Both primal and dual communities inhibit flow changes $|\Delta F|$ (colour coded) in the other community after the failure of a single edge (red) with unit flow in a given community. 
\textbf{e} We interpolate between primal and dual communities in a square grid of size $21\times10$ by tuning the weight $w_e$ of the horizontal edges or vertical edges (see a,b). The flow ratio $R$ reveals that failure spreading to the other community is largest for $w_\ell=1$. It decays for either type of community as measured by primal and dual Fiedler values $\lambda_2$ (crosses) and $\lambda_2^*$ (circles), respectively. The green line represents the median value and the shaded regions indicate the $25\%$ and $75\%$ quantiles. 
The symbol $c$ denotes a normalization constant used to improve comparison with $R$. 
\label{fig:LODF_ratio_SG}
}
\end{figure*}

Primal and dual communities both impede the spreading of failures and thus improve the robustness of complex networks as shown in Fig.~\ref{fig:connection_robustness}. We will now provide a more detailed and quantitative analysis of this connection for two important systems: flow networks and coupled oscillator networks.

We first consider linear flow networks using the theoretical framework introduced above. 
Robustness is quantified by a sensitivity factor,  measuring the response of the network flows $\vec F$ to a perturbation. 
As a perturbation, we add an inflow $\Delta P$ at a node $v_1$ and an outflow of the same amount at another node $v_2$. 
Here, we focus on the case where $v_1$ and $v_2$ are the two end nodes of an edge $e=(v_1,v_2)$ and treat the general case in the supplementary information. 
The source vector in the Poisson equation \eqref{eq:Poisson-primal} then changes as
\begin{equation}
    \vec P \rightarrow \vec P' = \vec P + \Delta P \, \matr{I}^\top \vec l_e 
\end{equation}
and $\vec{l}_e\in \mathbb{Z}^M$ is the indicator function for edge $e$, which is equal to one at the positions indicated by the subscript and zero otherwise. 
Inverting the discrete Poisson equation \eqref{eq:Poisson-primal}, we then find that the network flows change by the amount
\begin{equation}
    \Delta \vec F = \Delta P \, \matr{W} \matr{I} \matr{L}^\dagger \matr{I}^\top \vec l_e   ,
\end{equation}
where $\matr{L}^\dagger$ is the Moore-Penrose pseudoinverse of the primal graph Laplacian. We then define a sensitivity factor as the ratio of the flow change at edge $\ell$ and the perturbation strength $\Delta P$ as~\cite{Guo09,strake2018}
\begin{equation}
    \eta_{v_1,v_2,\ell} = \frac{\Delta F_\ell}{\Delta P} = 
    w_\ell \vec l_\ell^\top  \matr{I} \matr{L}^\dagger \matr{I}^\top \vec l_e.
    \label{eq:primal_ptdf}
\end{equation}
We note that the sensitivity factor is widely used in the context of power system security analysis, where it is referred to as a power transfer distribution factor~\cite{Guo09,strake2018}. Importantly, the sensitivity factor may also be used to simulate the failure of an edge $e=(v_1,v_2)$ by choosing the inflow $\Delta P$ accordingly (see supplementary information).

The sensitivity factor $\eta_{v_1, v_2, \ell}$ elucidates the relation between primal communities and network robustness~\cite{Manik2017}. In the supplementary information, we treat the limiting case of vanishing connectivity between the communities and show the following: If the edges $e$ and $\ell$ are in different communities, $\eta$ vanishes in the same way as the Fiedler value $\lambda_2$. If the edges $e$ and $\ell$ are in the same community, $\eta$ remains finite as $\lambda_2 \rightarrow 0$.

Remarkably, we can find an analogous description in the dual graph~\cite{ronellenfitsch_dual_2017,17lodf}. 
In equation \eqref{eq:KCL-general-solution}, we choose the particular solution as $\Delta F_{\rm part} = \Delta P \vec l_e$. 
We can then compute the cycle flows $\vec f$ from equation \eqref{eq:Poisson-dual} and substitute the result into equation \eqref{eq:F-Cf} to obtain the change of network flows~\cite{ronellenfitsch_dual_2017,17lodf}
\begin{equation}
    \Delta \vec F = - \Delta P \, \matr{C}  (\matr{L}^*)^\dagger \matr{C}^\top \matr{W}^{-1} \vec l_e  + \Delta P \, \vec l_e.
\end{equation}
The sensitivity factor for all edges $\ell \neq e$ thus reads
\begin{equation}
 \eta_{v_1,v_2,\ell} =  - \frac{1}{w_e} \vec{l}_\ell^\top\matr{C}(\matr{L}^*)^\dagger \matr{C}^{\top} \vec{l}_e.
  \label{eq:ptdf_dual}
\end{equation}
We see that the dual Laplacian $\matr{L}^*$ contributes to the sensitivity factor $\eta_{v_1,v_2,\ell}$ in exactly the same way as the primal Laplacian $\matr{L}$ in equation \eqref{eq:primal_ptdf}. 
Hence, we conclude that primal and dual community structures determine network flows in an equivalent manner. If the edges $e$ and $\ell$ are in different dual communities, $\eta$ will vanishes proportional to the dual Fiedler value $\lambda_2^*$. If the edges $e$ and $\ell$ are in the same community, $\eta$ remains finite in the limit $\lambda_2^* \rightarrow 0$.

We now quantify this effect. To analyse the impact of a community structure, we consider a square lattice with tunable edge weights. We either reduce the edge weights $w_{\ell}$ across the boundary, i.e~in the cut-set, to induce a primal community structure, or we increase the edge weights $w_{\ell}$ along the boundary, i.e.~in the cut-path, to induce a dual community structure (Fig~\ref{fig:LODF_ratio_SG}a,b).
We then consider an inflow and simultaneous outflow $\Delta P$ at two nodes $v_1$ and $v_2$, respectively, that are connected via an edge $e=(v_1,v_2)$. We then compare the resulting flow changes in the same (S) and the other (O) community as the given edge $e$. To this end, we evaluate the ratio of flow changes $R(e,d)$ in the two communities at a given distance $d$ to the trigger edge $e$~\cite{Kaiser2019}
\begin{equation}
    R(e,d)=\frac{\langle|\Delta F_{k} | \rangle_d^{k\in \text{O}}}{\langle
  |\Delta F_{r} |
  \rangle_d^{r\in \text{S}}}=\frac{\langle|\eta_{v_1,v_2,k} | \rangle_d^{k\in \text{O}}}{\langle
  |\eta_{v_1,v_2,r} |
  \rangle_d^{r\in \text{S}}}.
\end{equation}
Here, $\langle \cdot \rangle_d^{\ell \in C}$ denotes the average over all edges $\ell$ in a community $C$ at a given distance $d$ to the trigger edge $e$. To be able to neglect the effect of a specific edge and the distance, we average over all possible trigger edges $e$ and distances $d$ to arrive at the mean flow ratio
\begin{equation}
  R = \langle R(e,d)\rangle_{e,d}.
    \label{eq:flow_ratio}
\end{equation}
The mean flow ratio ranges from $R\approx 0$ if the other module is weakly affected, i.e. there is a strong community effect, to $R\approx 1$ if there is no noticeable effect. We note that $R$ describes flow changes after perturbations in the in- and outflows as well as flow changes as a result of the complete failure of edges (see supplementary information).

Figure~\ref{fig:LODF_ratio_SG} illustrates that both primal and dual communities suppress flow changes in the other community. The mean flow ratio $R$ decays for either community structure. In particular, this decay is well-captured by the Fiedler value of the primal ($\lambda_2$) and the dual ($\lambda_2^*$) graph.

\begin{figure*}[tb]
\begin{center}
  \includegraphics[width=1.0\textwidth]{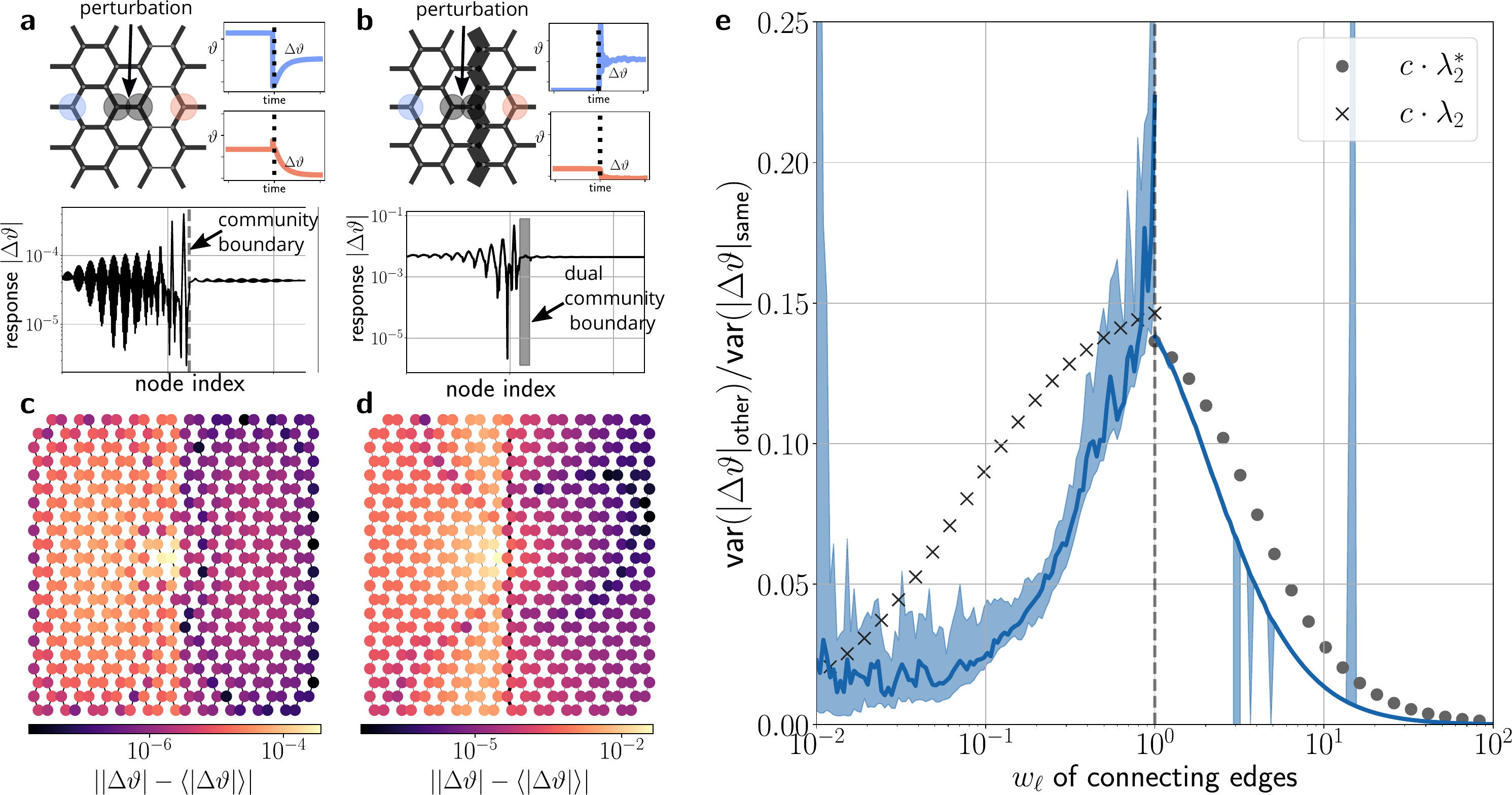}
\end{center}
\caption{\textbf{Suppression of failure spreading in oscillator networks.} 
\textbf{a,b} We analyse the response of a network of phase oscillators \eqref{eq:kuramoto} to a localised perturbation. We consider a honeycomb lattice with (\textbf{a}) a primal community structure induced by weak connectivity across the community boundary and (\textbf{b}) a dual community structure induced by strong connectivity along the boundary. After the perturbation the oscillators relax back to a phase-locked state with phases shifted by $\Delta \vartheta_i$.
\textbf{c,d} We find that the response $|\Delta \vartheta_i|$ is strongly suppressed in the non-perturbed community.
\textbf{e} The overall network response is quantified by the variance of the response within a network community \eqref{eq:phase-var}. The response in the non-perturbed community is strongly suppressed -- the more pronounced the community structure, the stronger the attenuation. 
The existence of a primal or dual community structure is indicated by primal and dual Fiedler values $\lambda_2$ (crosses) and $\lambda_2^*$ (circles), respectively. The blue line represents the median value and the shaded regions indicate the $25 \%$ and $75\%$ quantiles for different random realisations of the natural frequencies $\omega_i$. 
The symbol $c$ denotes a normalization constant used to improve comparison with the response.
\label{fig:Kuramoto}
}
\end{figure*}

These findings are not restricted to linear flow networks, but hold for all diffusively coupled networked systems. We illustrate this effect for a network of second-order phase oscillators that arises in the analysis of electric power grids \cite{motter2013spontaneous,dorfler_synchronization_2013} or mechanically coupled systems \cite{dorfler2014synchronization}
and as a generalization of the celebrated Kuramoto model \cite{acebron2005kuramoto}. The phase $\vartheta_i(t)$ of each oscillator $i=1,\ldots,N$ evolves according to
\begin{equation}
    M_i \ddot \vartheta_i + D_i \dot \vartheta_i = \omega_i + \sum_{j=1}^N w_{ij} \sin(\vartheta_j - \vartheta_i),
    \label{eq:kuramoto}
\end{equation}
where $M_i$ is the inertia and $D_i$ the damping of the $i$th oscillator. To analyse the impact of community structures, we consider a honeycomb lattice with tunable edge weights, with either low  weights $w_{ij} \le 1$ across the boundary or high weights $w_{ij} \ge 1$ along the boundary. The weights of all remaining edges are set to $w_{ij}=1$ and $w_{ij}=0$ if no edge exists between nodes $i$ and $j$.

We now investigate how the steady states of such a network react to a localized perturbation near the community boundary (Fig.~\ref{fig:Kuramoto}a,b). The oscillators relax to a phase-locked state after a short transient period, but the steady-state phases are shifted by an amount $\Delta \vartheta_i$. We recall that a global phase shift is physically irrelevant and is henceforth discarded. The response $|\Delta \vartheta_i|$ crucially depends on the location of the oscillator -- being strongly suppressed across the community boundary (Fig.~\ref{fig:Kuramoto}a-d). To evaluate the impact of the network structure, we quantify the overall network response by the variance of the phases within a community $C$,
\begin{equation}
    \mbox{var}\left( |\Delta \vartheta_i|_C  \right) 
    = \sum_{i \in C} |\Delta \vartheta_i|^2
     - \left( \sum_{i \in C} |\Delta \vartheta_i| \right)^2.
     \label{eq:phase-var}
\end{equation} 
This overall response is generally suppressed in the non-perturbed community, for primal as well as for dual communities. The more pronounced the community structure, the stronger the suppression of the response (Fig.~\ref{fig:Kuramoto}e).
We note that for the current example some differences exist between primal and dual communities. In particular,  statistic fluctuations are larger in the case of primal communities. 

We conclude that the impact of community boundaries, both primal and dual, extends to all diffusively coupled networked systems. 
Our finding can be further substantiated by a linear response analysis \cite{Manik2017}, which highlights the structural similarity to linear flow networks. 
Furthermore, we note that related phenomena were observed for models of information diffusion in networks of different modularity \cite{nematzadeh_optimal_2014}. This finding is closely related, as the diffusion model includes an averaging over all adjacent nodes in the network.

\section*{Discussion}

We have introduced a way to define and identify dual communities in planar graphs. We demonstrated that both primal and dual community structures emerge as different phases of optimized networks -- whether the one or the other is realized in a given optimal network depends on the degree of fluctuations. In addition to that, both types of communities have the ability to suppress failure spreading. They are thus optimized to limit the effect of edge failures or other perturbations.

An important difference between primal and dual communities is the fact that the former are based on a weak connectivity, while dual communities require a strong connectivity. This has significant consequences for supply networks such as power grids. Several approaches have been discussed to limit the connectivity of power grids to prevent the spreading of cascading failures. This includes concepts of microgrids \cite{lasseter2002microgrids} as well as intentional islanding \cite{mureddu2016islanding} or tree-partitioning \cite{bialek2021tree,zocca2021spectral}.
However, future power grids will require more, not less connectivity to transmit renewable energy over large distances ~\cite{Schlachtberger2017,trondle2020trade}. 
Dual communities might resolve this conundrum, as they prevent failure spreading from one community to the other one, without limiting the network's ability to transmit energy. 
This is in stark contrast to primal communities that limit failure spreading from one community to the other one, but also supply. Thus, the construction of dual communities may also serve as a strategy against failure spreading, in line with other ideas brought forward recently~\cite{Kaiser2019}.

Dual communities may be detected using the same techniques as for primal communities once the dual graph is constructed. We here focus on classical spectral methods based on the graph Laplacian, as this matrix naturally arises in the study of graph duality and linear flow networks. By now, numerous algorithms for community detection have been developed that outperform spectral methods depending on the respective application \cite{ newman_modularity_2006,peixoto2014hierarchical,ghasemian2019evaluating}. All these algorithms can be readily applied to the dual graph. A short analysis for a selected example is provided in the supplementary information. 
One challenge remains for the generalization of this approach. For planar graphs, the dual is constructed by a straightforward geometric procedure. For non-planar graphs, a geometric analysis is much more involved \cite{modes2016extracting}. A dual can be constructed algebraically by choosing a basis of the cycle space. However, there is no distinguished basis such that the algebraic dual is not unique. The detailed analysis of community boundaries, in particular the inequality \eqref{eq:lambda-bound-dual}, may provide an alternative route to generalize the definition network communities. For instance, one may choose a decomposition to minimize the dual topological connectivity $\mu_2^*$.

Finally, we note that other approaches have been put forward to generalize the definition of network communities beyond the paradigm of strong mutual connectivity.
For instance, communities can be defined in terms of the similarity of the connectivity of nodes (see, e.g.~ \cite{moore2011active,garcia2018applications}) or from spreading processes \cite{schaub2014structure}.
The graph dual approach presented here emphasizes the role of the community boundaries, both in the definition of the community structure and in its impact on spreading processes and network robustness. 
Furthermore, graph duality provides a rigorous algebraic justification for our generalization of community structures.

\section*{Methods}

\footnotesize

\subsection*{Global cascade model}

In Fig.~\ref{fig:connection_robustness}, we show results from a classic model of global cascades. The state of each node $i=1,\ldots,N$ in time step $t$ is denoted as $s_i(t) \in \{0,1\}$, encoding healthy/operational and infected/faulty, respectively. A node becomes infected/faulty in time step $t+1$ if the weighted average of the neighboring nodes exceeds a threshold $\phi_i$:
\begin{align}
   s_i(t+1) = \left\{ \begin{array}{l l l}  
   1 & \mbox{ if } &
   \frac{\sum_k w_{ik} s_k}{\sum_k w_{ik}} > \phi_i \\
   1 & \mbox{ if } & s_i(t) = 1\\
   0 & & \mbox{else}.
   \end{array}
   \right.
\end{align}
This model is iterated until no further changes of the node states occur. 

We simulate this model on a square lattice with inhomogeneous edge weights. 
A fraction $p_e=0.8$ of edges connecting the center nodes of the lattice with its nearest neighbors is selected at random. The weight of these edges is set to $w_{ij} = w_\ell$, where $w_\ell$ is a tunable parameter, while all other edges have weight $w_{ij} = 1$. 
At time $t=0$, we choose a fraction $\rho_0=0.05$ of all nodes in the left part and set them to state $1$, while all other nodes are in state $0$. For each value of the parameter $w_\ell$, we repeat the simulation for 1000 random initial conditions and record the fraction of nodes in state $1$, denoted as $\rho_\infty$.

\subsection*{Creation of dual graphs: planar networks}

In this manuscript, we mostly restrict our analysis to planar, connected graphs.
A graph $G=(V,E)$ with vertex set $V$ and edge set $E$ is called planar if it may be drawn in the plane without two edges crossing~\cite{Dies10}.
For a plane graph $G$, it is straightforward to establish a duality to another graph, referred to as the plane dual or simply dual graph and denoted as $G^*$.
The dual graph is constructed using the cycles of graph $G$ where a cycle is defined to be a path that starts and ends in the same vertex consisting of otherwise distinct vertices.
For a graph with $M$ edges and $N$ nodes, these cycles form the graph's cycle space of dimension $N^*=M-N+1$.
A particular basis of this space is given by the faces of the plane embedding, such that the dual graph $G^*=(V^*,E^*)$ has a vertex corresponding to each face.
Two dual vertices $v_1^*$ and $v_2^*$ are connected by a dual edge $e^*=(v_1^*,v_2^*)\in E^*(G^*)$ if the two corresponding cycles share an edge.
For a weighted graph, the edge weight of the dual edge is chosen to be the inverse of the corresponding edge shared by the two cycles.
Furthermore, we adopt the following convention; if two cycles share $k$ edges $e_1,..,e_k$ with weights $w_1,...,w_k$, we lump them together into a single dual edge $e^*$ with edge weight $w^*=\sum_{i=1}^k w_i^{-1}$ thus avoiding multi-edges in the dual graph and refer to this model as the reduced dual graph.
Note that the definition of the edge-cycle incidence matrix $\matr{C}$ needs to be adjusted for the reduced dual graph.

\subsection*{Creation of dual graphs: non-planar networks} For non-planar networks, the basis of the cycle space may no longer be uniquely determined based on the graph's embedding. Different basis choices result in different dual graphs. When calculating the dual graph of the non-planar European topology shown in Figure~\ref{fig:Optimal_supply_networks}k-m, we used the graphs' minimum cycle basis to create the dual graph.

\subsection*{Hierarchical decomposition of dual graph}
We assign $m$ hierarchy levels based on repeated spectral bisection of the dual graph using the following procedure: 
\begin{enumerate}
    \item Assign dual communities to the graph by making use of the Fiedler vector $\vec{v}_2^*$ of the dual graph $G^*$
    \item Identify the edges that lie on the boundary between the two communities by checking for edges in the primal shared by faces corresponding to dual nodes of both communities
    \item Remove the boundary edges from the graph thus creating two primal subgraphs $G_1$ and $G_2$
    \item Repeat the process $m$ times
\end{enumerate}

\subsection*{Building supply networks with fluctuating sources} 

Our framework extends the fluctuating sink model proposed by Corson~\cite{Corson2010} where a single, fluctuating source supplies the remaining network. To this end, we consider a linear flow network with sources and sinks attached to the nodes and model the sinks as Gaussian random variables $P\in\mathcal{N}(\mu,\sigma)$. In contrast to previous work, we consider multiple sources, $N_s$ in number, whose statistics can be derived from the statistics of the sinks due to the fact that the in- and outflows at the nodes need to sum to zero (see supplementary information). We then add additional fluctuations to the sources that are built using Dirichlet random variables $X_i\sim \operatorname{Dir}(\alpha)$. The fluctuations are constructed such that they only influence the statistics of the sources and their variance is tuned by a single parameter $\alpha$. To be able to tune the influence of this additive noise variable, we introduce a scale parameter $K\in\mathbb{R}$. The inflow at a source at a given point in time is then given by (see supplementary information)
\begin{equation}
    P_{s_i} = -\frac{1}{N_s}\sum_{i=N_s+1}^N P_i + K  \left(\frac{1}{N_s}-X_i\right),
\end{equation}
where $P_i$ are the outflows at the sinks. Here, we arranged the node order such that the sources have indices $1,\ldots,N_s$ and sinks are numbered as $N_s+1,\ldots,N$. To produce Figure~\ref{fig:Optimal_supply_networks}, we considered a network with $N_s=2$ and fix the parameters of the Gaussian distribution as $\mu = -1, \sigma = 0.1$.
The scale parameter is set to $K=500$ and the parameter $\alpha$ controlling the statistics of the Dirichlet distribution is varied in the interval
$\alpha\in[10^{-2},10^4]$, thus changing the variance of the Dirichlet variables $\sigma^2_D=K^2\frac{(N_s-1)}{N_s^2(N_s\alpha+1)}$ (see supplementary information).

\subsection*{Analysis of power grid datasets}

The networks shown in Figure~\ref{fig:Optimal_supply_networks}f-g were determined using the open energy system model 'PyPSA-Eur' cost-optimizing the transmission grid for different levels of carbon-dioxide emission reductions with respect to the emission levels in 1990. For each target carbon-dioxide emission reduction level, the network is optimized over for an entire year with the weather conditions of 2013 and 2-hourly resolution (see Ref.~\cite{horsch_2018} for further details on the optimization model). To analyse the network topology, we set the weight $w_\ell$ of a line $\ell$ to the maximal apparent power that can flow through it. Note that this is different from weighting the line by its line susceptance and allows us to also incorporate high-voltage DC lines. To determine the level of fluctuating renewables shown in Figure~\ref{fig:Optimal_supply_networks}f-g, we calculate the share of the total annual generation in the entire system that is produced by fluctuating renewables. To this end, we assume that the following technologies are fluctuating renewable energy sources: offshore wind AC, offshore wind DC, onshore wind, run-of-the-river hydroelectricity (ror) and solar. In Supplementary Figures~S3 and S4 we show as an example the generation for two months and carbon emission reduction levels over time and on the network level.

\subsection*{Data availability}

The topology of the Central European power grid have been extracted from the open European energy system model PyPSA-Eur~\cite{horsch_2018}, which is fully available online~\cite{horsch_zenodo}.

Leaf data was provided by the authors of Ref.~\cite{ronellenfitsch_topological_2015} and is available from the respective authors upon request. The leaf venation networks are based on microscopic recordings. Edge conductivities $w_{ij}$ are assumed to scale with the radius $r_{ij}$ of the corresponding vein $(i,j)$ as $w_{ij}\varpropto r_{ij}^4$ according to the Hagen-Poisseuille law, see Ref.~\cite{coomes_scaling_2008} for a detailed discussion. We used the radius in pixels at a resolution of 6400 dpi. 

The data generated in this study (effective topology of power grid networks and selected leaf venation networks) as well as essential computer code for data processing have been deposited in a Zenodo repository \cite{zenodo-dualcoms}.

\subsection*{Code availability.} 

Computer code is available on github \cite{github-dualcoms} with the specific version used in this publication being archived at Zenodo \cite{zenodo-dualcoms}.

\noindent{\textbf{Acknowledgements}}\\
We thank Torsten Eckstein for providing some of the digitized leaf networks, Tom Brown and Fabian Neumann for providing us with the optimized power grids and Eleni Katifori for helpful discussions. 
We gratefully acknowledge support from the German Federal Ministry of Education and Research (BMBF) via the grant ``CoNDyNet2'' with grant no. 03EK3055B, the Helmholtz Association via the grant ``Uncertainty Quantification -- From Data to Reliable Knowledge (UQ)'' with grant no.~ZT-I-0029 and the Deutsche Forschungsgemeinschaft (DFG, German Research Foundation) with grant No. 491111487. \\ \\
\noindent\textbf{Author Contributions}\\
\normalsize D.W. conceived research and acquired funding. F.K. and D.W. designed research. F.K. carried out all numerical simulations. F.K. and P.C.B. evaluated the results and designed the figures. All authors contributed to discussing the results and writing the manuscript.\\ \\
 \normalsize\textbf{Competing Interests} \\
The Authors declare no Competing Financial or Non-Financial Interests.\\ \\

\textbf{Correspondence and requests for materials} should be addressed to Dirk Witthaut~(email: d.witthaut@fz-juelich.de).

\end{document}


\title{Supplementary Information for \\ Dual communities in spatial networks}

\author{Franz Kaiser}
    \affiliation{Forschungszentrum J\"ulich, Institute for Energy and Climate Research (IEK-STE), 52428 J\"ulich, Germany}
    \affiliation{Institute for Theoretical Physics, University of Cologne, K\"oln, 50937, Germany}
\author{Philipp C. Böttcher}
  \affiliation{Forschungszentrum J\"ulich, Institute for Energy and Climate Research (IEK-STE), 52428 J\"ulich, Germany}
\author{Henrik Ronellenfitsch}
    \affiliation{Physics Department, Williams College, 33 Lab Campus Drive, Williamstown, MA 01267, U.S.A.}
    \affiliation{Department of Mathematics, Massachusetts Institute of Technology, Cambridge, MA 02139, U.S.A.}
\author{Vito Latora}
    \affiliation{School of Mathematical Sciences, Queen Mary University of London, London E1 4NS, UK}
    \affiliation{Dipartimento di Fisica ed Astronomia, Universit{\`a} di Catania and INFN, 95123 Catania, Italy}
    \affiliation{Complexity Science Hub Vienna, 1080 Vienna, Austria}
\author{Dirk Witthaut}%
    \affiliation{Forschungszentrum J\"ulich, Institute for Energy and Climate Research (IEK-STE), 52428 J\"ulich, Germany}
    \affiliation{Institute for Theoretical Physics, University of Cologne, K\"oln, 50937, Germany}

\date{\today}

\maketitle

\onecolumngrid
    \noindent This Supplementary Material contains five Supplementary Notes and six Supplementary Figures.
   
\newpage
\widetext
\tableofcontents

\newpage
\clearpage
\renewcommand{\figurename}{Supplementary Figure}
\section*{Supplementary Figures}

\begin{figure*}[ht!]
\begin{center}
  \includegraphics[width=1.\textwidth]{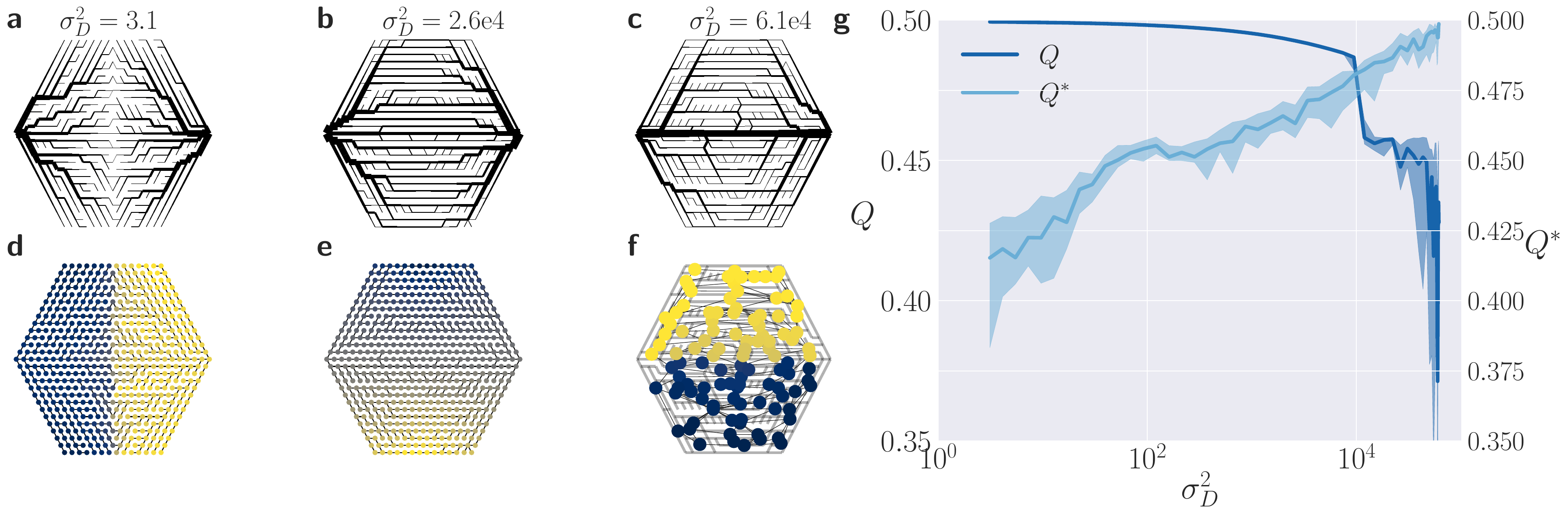}
\end{center}
\caption{\textbf{Primal and dual communities emerge naturally in optimal supply networks.} 
\textbf{a-c,} We consider a triangular lattice with two fluctuating sources located at the leftmost and the rightmost corner. 
The remaining nodes correspond to sinks, whose demand fluctuates following an iid normal distribution. 
The sources balance the demand in total, but the spatial distribution varies according to a Dirichlet distribution introducing additional fluctuations with variance $\sigma_D^2$.
We then determine the optimal network structure and edge weights to minimise the average dissipation~\cite{Corson2010}. 
We observe a transition from (\textbf{a}) a primal community structure at low fluctuations to (\textbf{c}) a dual community structure at high fluctuations.
\textbf{d-f,} Communities can be identified according to the Fiedler vector. The figure shows (\textbf{d,e}) the primal Fiedler vector and (\textbf{f}) the dual Fiedler vector in a colour map for the networks depicted in panels (\textbf{a-c}), respectively.
\textbf{g,} The transition from a primal to a dual community structure can be quantified either by the algebraic connectivity (cf.~main text) or by the modularity \cite{newman_modularity_2006}. The figure shows the modularity $Q$ of the primal community decomposition and the modularity $Q^*$ of the dual community decomposition as a function of the variance $\sigma_D^2$ of the source strength.
\label{fig:Optimal_supply_networks_mod} 
}
\end{figure*}

\begin{figure*}
\begin{center}
  \includegraphics[width=1.\textwidth]{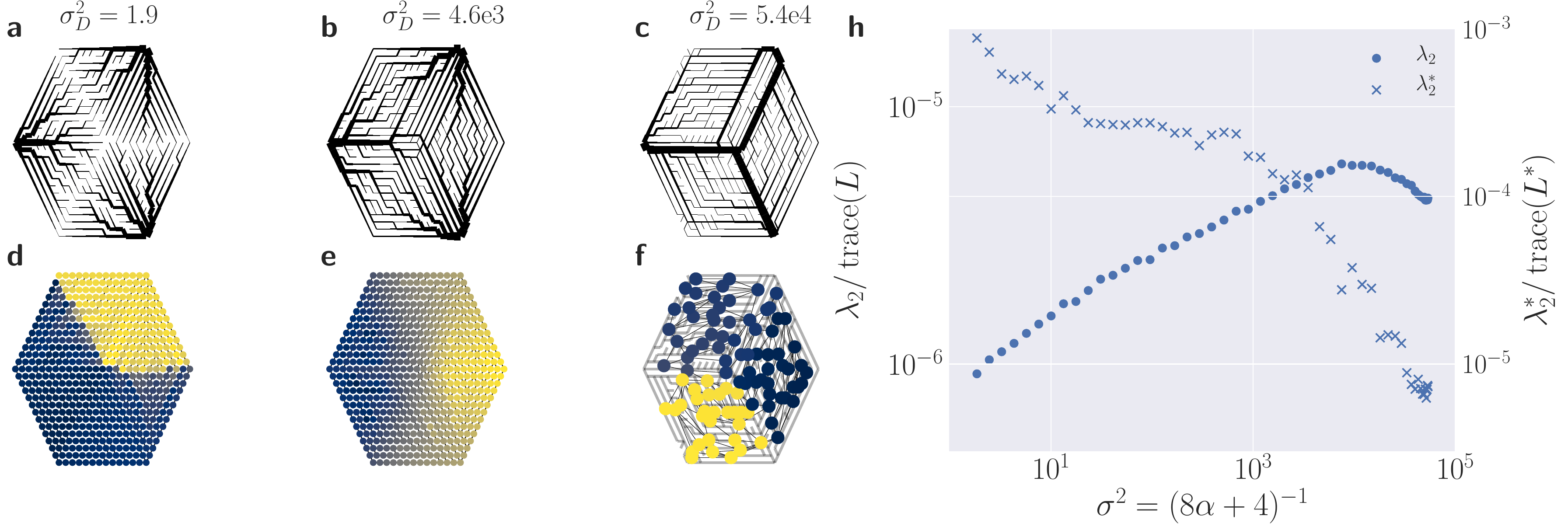}
\end{center}
\caption{\textbf{Primal and dual communities emerge naturally in optimal supply networks with three sources.} 
\textbf{a-c,} We consider a triangular lattice with three fluctuating sources located at the leftmost, the upper right and the lower right corner. The remaining nodes correspond to sinks, whose demand fluctuates following an iid normal distribution. The sources balance the demand in total, but the spatial distribution varies according to a Dirichlet distribution introducing additional fluctuations with variance $\sigma_D^2$. We then determine the optimal network structure and edge weights to minimise the average dissipation~\cite{Corson2010}. We observe a transition from (\textbf{a}) a primal community structure at low fluctuations to (\textbf{c}) a dual community structure at high fluctuations.
\textbf{d-f,} Communities can be identified according to the Fiedler vector. 
The figure shows (\textbf{d,e}) the primal Fiedler vector and (\textbf{f}) the dual Fiedler vector in a colour map for the networks depicted in panels (\textbf{a-c}), respectively.
\textbf{h,} The scaling of the corresponding primal ($\lambda_2$, circles) and dual ($\lambda_2^*$, crosses) Fiedler value confirms the transition.
\label{fig:Optimal_supply_networks_three_sources}
}
\end{figure*}

\begin{figure*}
\begin{center}
  \includegraphics[width=1.\textwidth]{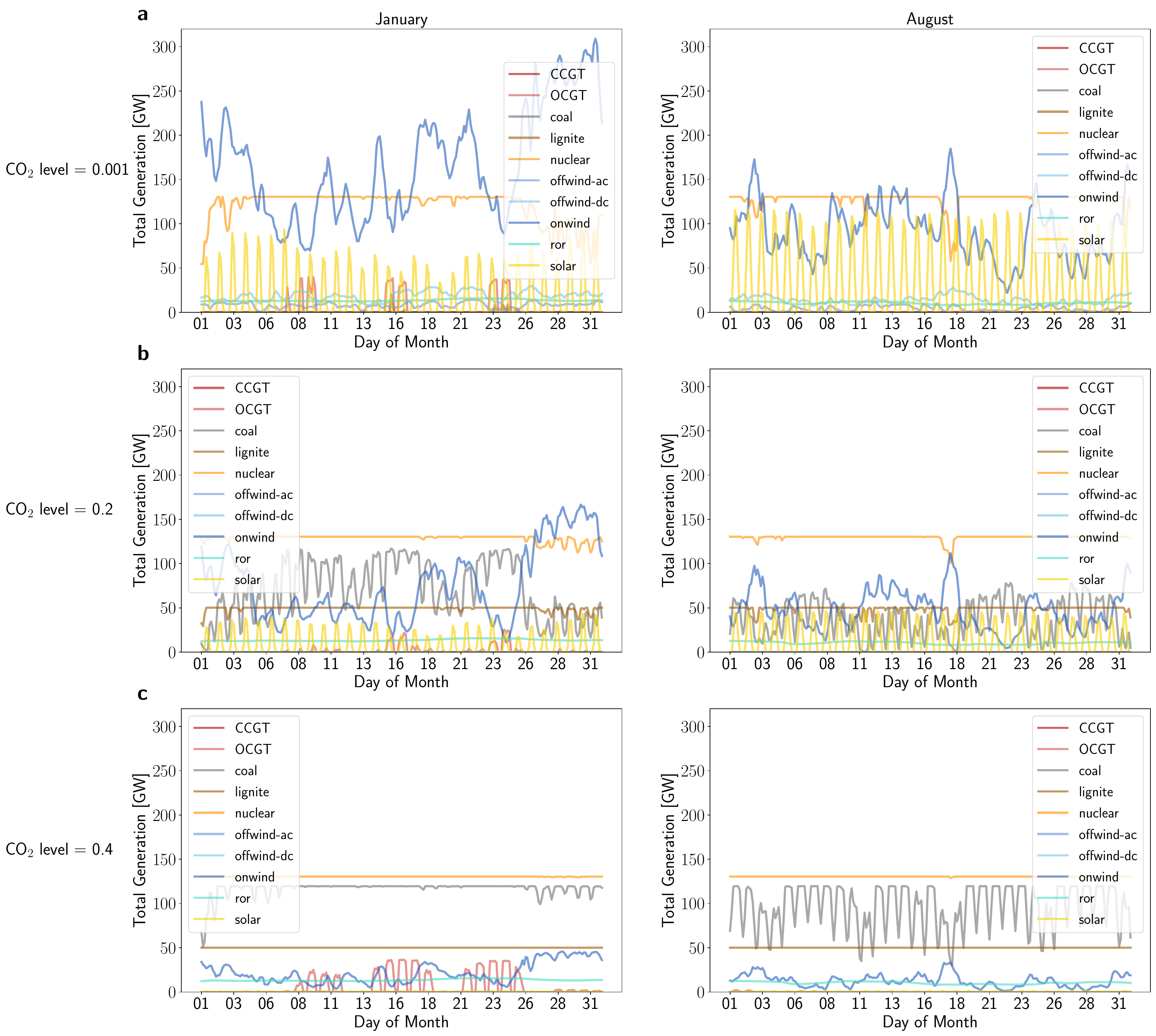}
\end{center}
\caption{\textbf{Optimized electric power systems simulated with the open energy system model 'PyPSA-Eur'.}
The model optimizes both the investments and the operation of the generation and transmission infrastructures to minimize the total system costs for a given CO$_2$ cap. The figure shows the total power generation by generator type for two exemplary months in winter (January, left) and summer (August, right) with a temporal resolution of 2 hours. Results are shown for three different values of the CO$_2$ level measured relative to the reference year 1990: (\textbf{a}) 0.001,  (\textbf{b}) 0.2, and 
(\textbf{c}) 0.4.
With decreasing levels of CO$_2$, generation shifts more and more towards fluctuating renewable power sources.  \label{fig:pypsa_illustration}
}
\end{figure*}

\begin{figure*}
\begin{center}
  \includegraphics[width=.8\textwidth]{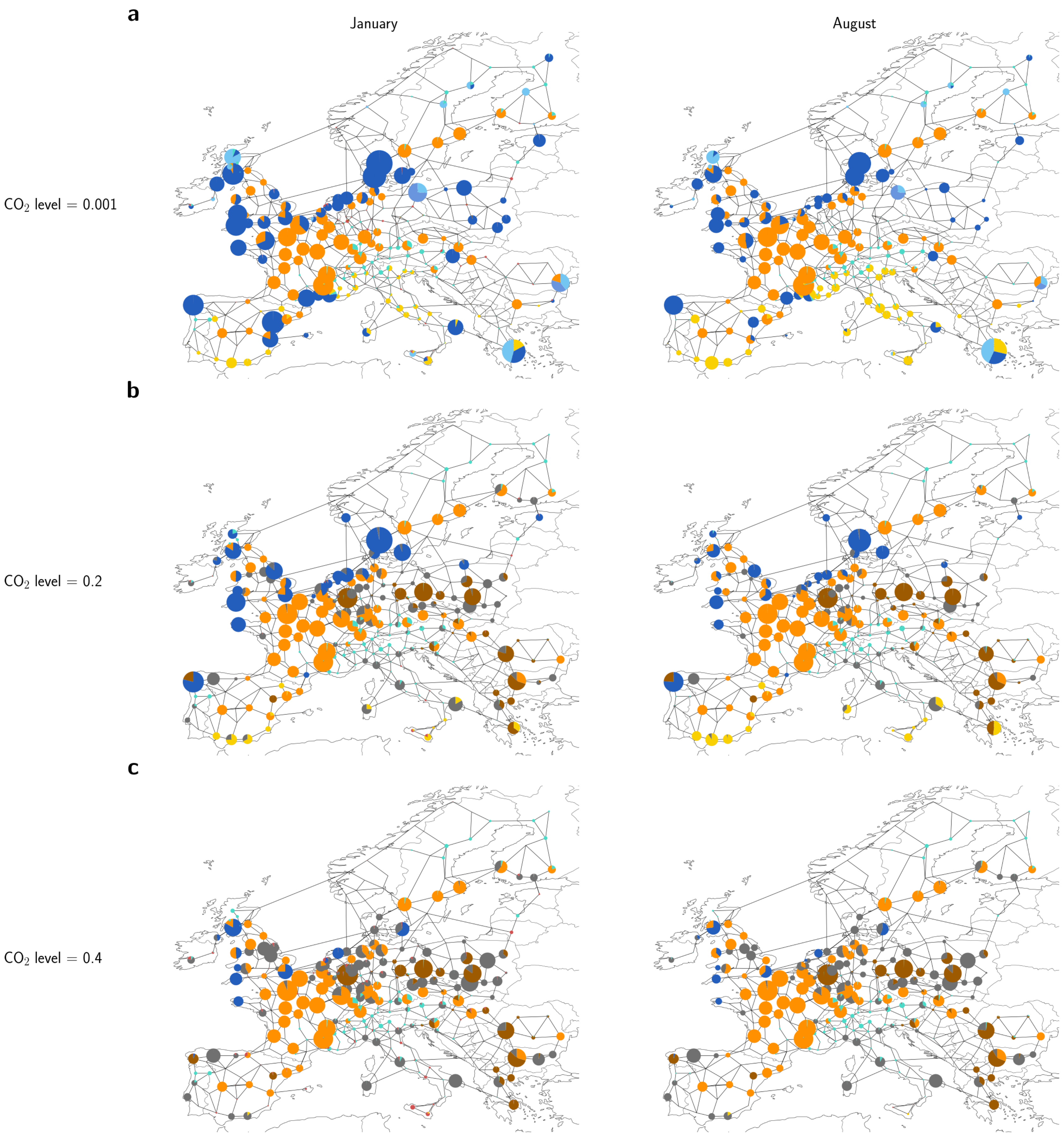}
\end{center}
\caption{
\textbf{Electric power systems and grids during the energy transition.}
The figure shows the change in power generation and transmission infrastructures as the CO$_2$ level is reduced. Power system layout and operation are optimized using the open energy system model 'PyPSA-Eur'. The figure shows the total generation per month and grid node for two exemplary months in winter (January, left) and summer (August, right) as well as the layout of the power transmission grid. Generation is shown in a pie chart, where the size is proportional to the total generation and the colour code indicates the generation type, cf.~Supplementary Figure \ref{fig:pypsa_illustration}. Results are shown for three different values of the CO$_2$ level measured relative to the reference year 1990: (\textbf{a}) 0.001,  (\textbf{b}) 0.2, and 
(\textbf{c}) 0.4. With decreasing levels of CO$_2$, generation shifts more and more towards fluctuating renewable power sources with pronounced regional differences: Wind power in predominantly generated along the coast, especially around the North Sea. Solar power is generated predominantly in Southern Europe.
\label{fig:pypsa_illustration_networks}
}
\end{figure*}

\begin{figure*}
\begin{center}
  \includegraphics[width=1.\textwidth, clip, trim = 1cm 2cm 6cm 2cm]{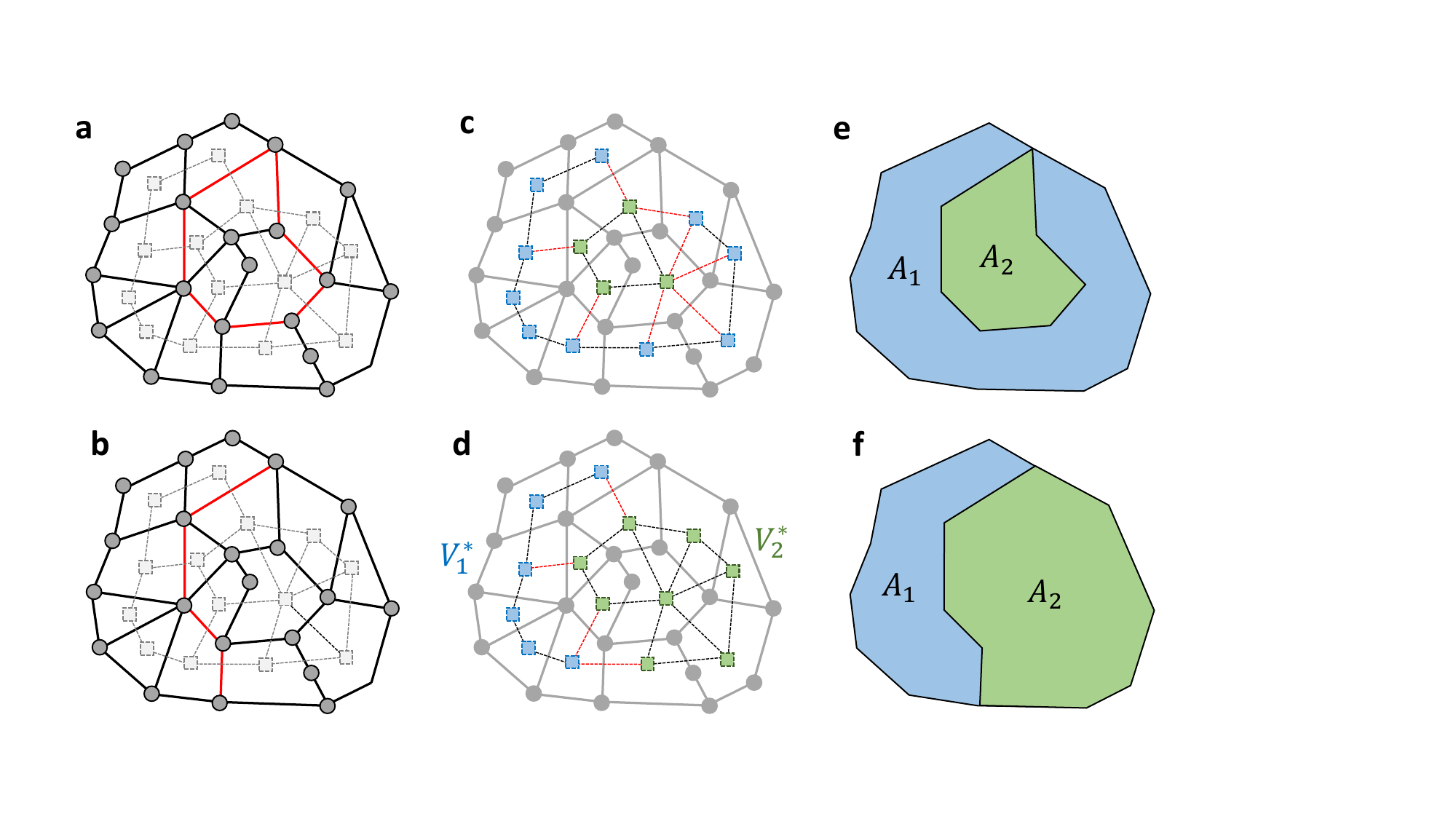}
\end{center}
\caption{\textbf{Definition of a cut-path.} 
\textbf{a,b,} A graph (solid lines) and its weak dual (dashed lines). The edge set of a cut-path $S$ is indicated in red colour.
\textbf{c,d,} The cut-path induces a cut of the dual graph into two connected vertex sets $V_1^*$ (blue) and $V_2^*$ (green). The set $S^*$ is a minimal cut set in the dual graph.  
\textbf{e,f,} The faces corresponding to the vertex sets $V_1^*$ and $V_2^*$ form two connected polygons $A_1$ and $A_2$ in the plane. 
\label{fig:dual-cut}
}
\end{figure*}

\begin{figure*}[h]
\begin{center}
  \includegraphics[width=1.\textwidth]{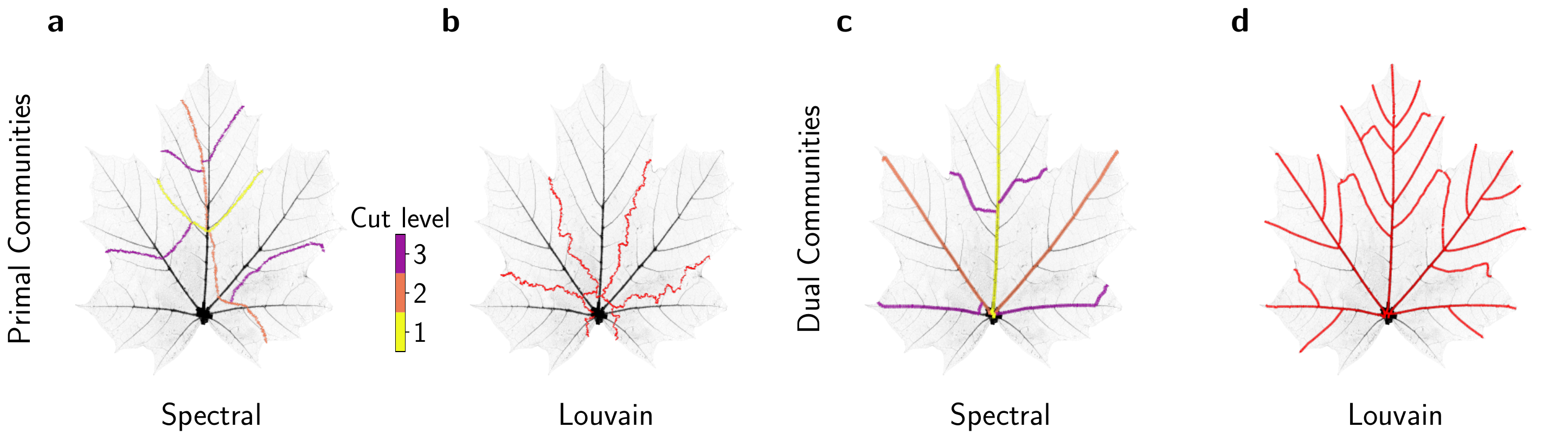}
\end{center}
\caption{\textbf{Different clustering methods in primal and dual graphs} 
Distinct primal and dual communities can be found using both spectral clustering and the \textit{Louvain Detection Algorithm}. 
The original (i.e. primal) leaf venation network of \textit{Acer platanoides} is shown in black in each panel.
The resulting hierarchical cuts for three cut levels employing spectral clustering using the primal graph Laplacian and the dual graph Laplacian can be seen in panel \textbf{a} and \textbf{c}.
The community structure revealed by the \textit{Louvain Detection Algorithm}, choosing a resolution parameter $\xi=0.01$, using the primal graph and dual graph are shown in panel \textbf{b} and \textbf{d}, respectively.
\label{fig:clustering_primal_dual}
}
\end{figure*}

\clearpage
\newpage

\section*{Supplementary Notes}

\refstepcounter{notecounter}
\subsection*{Supplementary Note \thenotecounter :~Graph duality}

In this note, we review the fundamental definitions and properties of dual graphs extending the main text.

\subsubsection*{Basic definitions}

Let $G=(V,E)$ be a plane, undirected graph with $N= |V|$ vertices and $L=|E|$ edges. We label all vertices consecutively as $(1,2,3,\ldots,N)$. Edges are either referred to in terms their respective endpoints or by consecutive labels as $(1,2,3,\ldots, L)$. We begin with unweighted graphs and include edge weights later. For each edge, we fix an orientation -- for instance to keep track about the direction of flows as discussed below. The topology of the graph and the orientations are encoded in the node-edge incidence matrix $\mathbf{I} \in \mathbb{R}^{L \times N}$ with entries~\cite{newman2010}
\begin{equation*}
   I_{\ell,n} = \left\{
   \begin{array}{r l}
      1 & \; \mbox{if edge $\ell$ starts at node $n$},  \\
      - 1 & \; \mbox{if edge $\ell$ ends at node $n$},  \\
      0     & \; \mbox{otherwise}.
  \end{array} \right.
\end{equation*}
The Graph Laplacian is defined as 
\begin{equation}
    \mathbf{L} = \mathbf{I}^\top \mathbf{I}.
    \label{eq:Laplacian-Primal-I}
\end{equation}
This definition coincides with the definition 
\begin{equation*}
    \mathbf{L} = \mathbf{D}-\mathbf{A},
\end{equation*}
in terms of the adjacency matrix $\mathbf{A} \in \mathbb{R}^{N \times N}$ with entries
\begin{equation*}
   A_{m,n} = \left\{
   \begin{array}{r l}
      1 & \; \mbox{if nodes $m$ and $n$ are connected},  \\
      0     & \; \mbox{otherwise}.
  \end{array} \right.
\end{equation*}
and the diagonal degree matrix $\mathbf{D} \in \mathbb{R}^{N \times N}$ with entries
\begin{equation*}
   D_{n,n} = \sum_m A_{n,m}.
\end{equation*}

A cycle in the graph can be described by a vector $\vec v \in \mathbb{R}^L$ with entries
\begin{equation*}
   v_{\ell} = \left\{
   \begin{array}{r l}
      + 1 & \; \mbox{if edge $\ell$ is part of the cycle},  \\
      - 1 & \; \mbox{if the reversed edge $\ell$ is part of the cycle},  \\
      0     & \; \mbox{otherwise}.
  \end{array} \right.
\end{equation*}
that satisfies 
\begin{equation}
   \mathbf{I} \vec v = \vec 0.     
\end{equation}
This property is easy to understand for flow networks: Here, the vector $\vec v$ describes a cycle flow which has neither sink or source such that  $\mathbf{I} \vec v = \vec 0$. The cycle flows constitute a vector space that coincides with the kernel of the edge incidence matrix $\mathbf{I}$. In plane graphs, the faces provide a distinguished basis of the cycle space. This basis is encoded in the cycle incidence matrix $\mathbf{C}$ with entries
\begin{equation*}
   C_{\ell,f} = \left\{
   \begin{array}{r l}
      1 & \; \mbox{if the edge $\ell$ is part of face $f$},  \\
      - 1 & \; \mbox{if the reversed edge $\ell$ is part of face $f$},  \\
      0     & \; \mbox{otherwise}.
  \end{array} \right.
\end{equation*}
For convenience, we fix the orientation of each face in the counter-clockwise direction.

We now proceed to the definition of the dual graph $G^* = (V^*,E^*)$. Two common definitions exist, differing in the treatment of the exterior of the graph.
\begin{enumerate}
    \item The vertex set $V^*$ consists of all faces of the primal graph, including the outer face. Then, the dual graph has $|V^*| = L-N+2$ vertices. 
    \item The vertex set $V^*$ consists of all faces of the primal graph, excluding the outer face. Then, the dual graph has $|V^*| = L-N+1$ vertices. This graph is typically referred to as a \emph{weak dual}.
\end{enumerate}
Two faces (=dual vertices) are adjacent if they share at least one edge. We note that two faces may share more than one edge, which can be included in two different ways. First, we keep every edge separately making the dual graph a multigraph. Alternatively, we can condense the multiedges into ordinary ones, which requires the introduction of edge weights. We will follow the latter definition as our focus is on flow networks which are weighted anyway.

We may now define a Laplacian for the dual graph.
Starting with the first definition that includes the outer face, we define, in analogy to Eq.~\eqref{eq:Laplacian-Primal-I}, the dual Laplacian as
\begin{equation}
    \mathbf{L}^* = \mathbf{C}^\top \mathbf{C} 
           \in \mathbb{R}^{(L-N+2) \times (L-N+2)} \, .
    \label{eq:Laplacian-Dual-C}
\end{equation}
As for the primal graph, we can decompose the Laplacian in terms of the adjacency and the degree matrix,
\begin{equation}
    \mathbf{L}^* = \mathbf{D}^* - \mathbf{A}^*.
    \label{eq:Laplacian-Dual-DA}
\end{equation}
The dual adjacency matrix $\mathbf{A}^* \in \mathbb{R}^{(L-N+2) \times (L-N+2)}$ has the entries
\begin{equation*}
   A^*_{f,g} = \mbox{number of edges shared by faces $f$ and $g$} \, .
\end{equation*}
and the degree matrix $\mathbf{D}^* = \in \mathbb{R}^{(L-N+2) \times (L-N+2)}$ is diagonal with entries
\begin{equation*}
   D^*_{f,f} = \sum_g A_{f,g}.
\end{equation*}

In the second definition, we exclude the exterior face. We may again proceed to define the Graph Laplacian via the cycle incidence matrix and write
\begin{equation}
     \mathbf{L}_g^* = \mathbf{C}^\top \mathbf{C} 
        \in \mathbb{R}^{(L-N+1) \times (L-N+1)}.
     \label{eq:Laplacian-dual-grounded}
\end{equation}
However, this matrix is \emph{not} a full Laplacian matrix as its row-sums are not guaranteed to vanish. Instead, $\mathbf{L}_g^*$  is a \emph{grounded} Laplacian, that is obtained from a full Laplacian matrix by removing one row and the corresponding column~\cite{miekkala_graph_1993}. Alternatively, we can define the Laplacian via the adjacency matrix as
\begin{equation}
    \mathbf{L}_i^* = \mathbf{D}^* - \mathbf{A}^*.
    \label{eq:Laplacian-Dual-DA-i}
\end{equation}
The adjacency matrix $\mathbf{A}^*$ is a $(L-N+1) \times (L-N+1)$ matrix has the
entries
\begin{equation}
   A^*_{f,g} = \mbox{number of edges shared by faces $f$ and $g$} \, .
   \label{eq:adjacency-dual-red}
\end{equation}
as before and the degree matrix $\mathbf{D}^* = \in \mathbb{R}^{(L-N+1) \times (L-N+1)}$ is diagonal with entries
\begin{equation*}
   D^*_{f,f} = \sum_g A_{f,g}.
\end{equation*}
Hence, the row-sums of $\mathbf{L}_i^*$ vanish as required for a Laplacian matrix. We note that the two matrices are related as
\begin{equation}
    \mathbf{L}_g^* = \mathbf{L}_i^* + \mathbf{D}^\circ \, ,
\end{equation}
where the diagonal matrix $\mathbf{D}^\circ$ summarizes the connectivity of the interior and exterior
\begin{equation}
    D^\circ_{f,f} = \mbox{number of edges shared between $f$ and the exterior face}.
\end{equation}

Finally, we want to briefly comment on the two alternative definitions of the dual graph and their respective advantages and disadvantages: 
\begin{itemize}
\item
The first definition including the exterior face has two closely related advantages making this definition convenient in fundamental analysis: (i)  The definitions \eqref{eq:Laplacian-Dual-C} and \eqref{eq:Laplacian-Dual-DA} of the graph Laplacian coincide. (ii) Every edge is part of exactly two faces. This property forms the basis of MacLane's planarity criterion.

\item
However, the first definition is rather inconvenient in the analysis of network structures. The cycle basis is overcomplete. The additional exterior face introduces ``artificial'' non-local connections in the dual, which are especially unsuitable for the definition and analysis of dual communities. For community detection, we will thus work only with the adjacency matrix \eqref{eq:adjacency-dual-red} and the Laplacian \eqref{eq:Laplacian-Dual-DA-i}.

\item
In the analysis of flow networks we can use both definitions alike. In the first case, we have one degree of freedom that remains to be fixed. A common convention is to ``ground'' the exterior face -- then the first definition reduces to the second one. 
\end{itemize}

\subsubsection*{Weighted Graphs I}

We now consider a weighted plane graph with edge weights denoted as $w_e \in \mathbb{R}_{>0}$ for all edges $e\in E$.
In the analysis of graph duality and flow networks, we can include these weights in two different ways: (i) directly in the incidence matrices and (ii) via separate diagonal weight matrices. 
We first treat option (i) and focus on the definition of the dual that includes the exterior face. 

Subsequently, we define the weighted node-edge incidence matrix 
$\mathbf{I} \in \mathbb{R}^{L \times N}$ as
\begin{equation*}
   I_{e,n} = \left\{
   \begin{array}{r l}
      +\sqrt{w_e} & \; \mbox{if edge $e$ starts at node $n$},  \\
      - \sqrt{w_e} & \; \mbox{if edge $e$ ends at node $n$},  \\
      0     & \; \mbox{otherwise}.
  \end{array} \right.
\end{equation*}
The Graph Laplacian is defined as before
\begin{equation*}
   \mathbf{L} = \mathbf{I}^\top \mathbf{I}.
\end{equation*}
with entries
\begin{equation}
    L_{mn}   =\left\{\begin{array}{l l }
      -w_{mn} & \; \mbox{if $m$ is connected to $n$},  \\
      \sum_{(m,k) \in E} w_{mk} & \; \mbox{if $m=n$},  \\
      0     & \; \mbox{otherwise}.
  \end{array} \right. 
  \label{eq:Laplacian}
\end{equation}
To ensure the defining equation for the cycle space, $\mathbf{I}^\top \, \mathbf{C} = 0$, the definition of the cycle edge incidence matrix must be generalized to
\begin{equation*}
   C_{\ell,f} = \left\{
   \begin{array}{r l}
      1/\sqrt{w_\ell} & \; \mbox{if the edge $\ell$ is part of face $f$},  \\
      - 1/\sqrt{w_\ell} & \; \mbox{if the reversed edge $\ell$ is part of face $f$},  \\
      0     & \; \mbox{otherwise}.
  \end{array} \right.
\end{equation*}
The dual Laplacian matrix is again given by the equivalent definitions \eqref{eq:Laplacian-Dual-C} and \eqref{eq:Laplacian-Dual-DA}, where the adjacency matrix is now given by
\begin{equation}
   A^*_{f,g} = \sum_{\mbox{edges $e$ shared by faces $f$ and $g$}} \frac{1}{w_e} \, .
   \label{eq:dual-weighted-adjacency}
\end{equation}
and the diagonal degree matrix is given by $D^*_{f,f} = \sum_g A_{f,g}$ as usual. We stress that the dual weights are inverse to the primal weights as expressed via Eq.~\eqref{eq:dual-weighted-adjacency}.

\subsubsection*{Weighted Graphs II}

A second option is to include the weights via an additional diagonal weight matrix
\begin{equation}
   \mathbf{W} = \mbox{diag}(w_1,w_2,\ldots,w_L) \in \mathbb{R}^{L\times L}
   \label{eq:weight-matrix}
\end{equation} 
while keeping the unweighted incidence matrix $\tilde{\mathbf{I}} \in \mathbb{R}^{L \times N}$ with entries
\begin{equation}
   \tilde I_{\ell,n} = \left\{
   \begin{array}{r l}
      1 & \; \mbox{if edge $\ell$ starts at node $n$},  \\
      - 1 & \; \mbox{if edge $\ell$ ends at node $n$},  \\
      0     & \; \mbox{otherwise}.
  \end{array} \right.
  \label{eq:incidence-unweighted}
\end{equation}
The Laplacian matrix of the primal graph is then defined by the definition
\begin{equation}
   \mathbf{L} = \tilde{\mathbf{I}}^\top \mathbf{W} \tilde{\mathbf{I}},
   \label{eq:Laplacian-weighted-IWI}
\end{equation}
which coincides with Eq.~\eqref{eq:Laplacian} as desired.

We now proceed to the dual, focusing on the second definition that excludes the exterior face. The cycle edge incidence matrix $\tilde{\mathbf{C}} \in \mathbb{R}^{L \times (L-N+1)}$ now remains unweighted with components 
\begin{equation}
   \tilde C_{\ell,f} = \left\{
   \begin{array}{r l}
      1 & \; \mbox{if the edge $\ell$ is part of face $f$},  \\
      - 1 & \; \mbox{if the reversed edge $\ell$ is part of face $f$},  \\
      0     & \; \mbox{otherwise}.
  \end{array} \right.
  \label{eq:cycle-in-unweighted}
\end{equation}
The dual grounded Laplacian \eqref{eq:Laplacian-dual-grounded} is recovered via the definition
\begin{equation}
    \mathbf{L}_g^* = \tilde{\mathbf{C}}^\top \mathbf{W}^{-1} \tilde{\mathbf{C}}.
    \label{eq:Laplacian-weighted-dual-CWC}
\end{equation}
Comparing to the definition of the primal Laplacian \eqref{eq:Laplacian-weighted-IWI}, we again see that dual weights are inverse to primal weights.

\clearpage

\refstepcounter{notecounter}
\subsection*{Supplementary Note \thenotecounter:~Algebraic and topological connectivity}

The correspondence of algebraic and topological connectivity is well established for primal graphs. Here, we briefly review this correspondence and extend it to the dual graph.

\subsubsection*{Primal Graph}

Let $G=(V,E,\mathbf{W})$ be a weighted connected graph and $\mathbf{L}$ its Laplacian. The Laplacian has one vanishing eigenvalue $\lambda_1 = 0$ with corresponding eigenvector $\vec v_1 \propto \vec 1 = (1,1,\ldots,1)^\top$. Since $\mathbf{L}$ is positive semi-definite~\cite{newman2010}, we can write the eigenvalues in an ordered way
\begin{equation}
    0 = \lambda_1 \le \lambda_2 \le \cdots \le \lambda_N.
\end{equation}
The second eigenvalue $\lambda_2$, the Fiedler value, provides a measure of the algebraic connectivity of a graph.
If it is small, we have a  pronounced community structure.

We now relate the algebraic connectivity $\lambda_2$ to the topological connectivity across a cut-set of the graph. To this end, we first express $\lambda_2$ via the variational principle as 
\begin{equation}
    \lambda_2 = \min_{\vec v \perp \vec 1} \frac{\vec v^\top \mathbf{L} \vec v}{
                       \vec v^\top \vec v} \, .
\end{equation}
We can then find an upper bound by choosing an appropriate trial vector for $\vec v$. To this end, we decompose the vertex set into two parts $V_1$ and $V_2 = V \backslash V_1$ and define a trial vector with entries
\begin{equation}
   v_{n} = \left\{
   \begin{array}{r l}
      + \left( \frac{N_2}{N \, N_1} \right)^{1/2}  & \; \mbox{if} \;  n \in V_1,  \\
       - \left( \frac{N_1}{N \, N_2} \right)^{1/2}  & \; \mbox{if} \;  n \in V_2,
  \end{array} \right.
  \label{eq:trial-vector}
\end{equation}
where $N_1 = |V_1|$ and $N_2 = |V_2|$. Then we have
\begin{align}
    \lambda_2 &\le \frac{\vec v^\top \mathbf{L} \vec v}{
                       \vec v^\top \vec v} \nonumber \\
                    &= \frac{N_1+N_2}{N_1 \, N_2} \sum_{m\in V_1, n \in V_2} A_{m,n}   
                       \nonumber \\
                    &= \frac{N_1+N_2}{N_1 \, N_2}  \sum_{e \in S} w_e , 
	\label{eq:fiedler-primal-bound}
\end{align}
where $S = \{  (m,n) \in E | m \in V_1, n \in V_2 \}$ is the cut-set corresponding to the cut $V = V_1 + V_2$.

We note that the expression \eqref{eq:fiedler-primal-bound} is not just an upper bound but gives $\lambda_2$ up to linear order in the topological connectivity. To make this more precise, assume that we start from a graph that is disconnected into two parts
$V_1$ and $V_2 = V \backslash V_2$. We label the vertices as $V_1 = \{1,\ldots,N_1\}$ and $V_2 = \{N_1+1,\ldots,N_2+N_2\}$. The adjacency matrix then assumes a block form
\begin{equation*}  
    \mathbf{A}^{(0)} = 
         \begin{pmatrix}
         \mathbf{A}_{11} & \mathbf{0} \\
         \mathbf{0} & \mathbf{A}_{22}
         \end{pmatrix},
\end{equation*}  
where $\mathbf{A}^{(0)}_1 \in \mathbb{R}^{N_1 \times N_1}$ and $\mathbf{A}^{(0)}_2 \in \mathbb{R}^{N_2 \times N_2}$. In this limit we have $\lambda_2^{(0)} = 0$. 
Now add some weak connections between the two modules of the graph. Then, the adjacency matrix can be written as 
\begin{equation*}  
    \mathbf{A} = \mathbf{A}^{(0)} + \mathbf{A}^{(1)} =
         \begin{pmatrix}
         \mathbf{A}_{11} & \mathbf{0} \\
         \mathbf{0} & \mathbf{A}_{22}
         \end{pmatrix} + 
         \begin{pmatrix}
         \mathbf{0} & \mathbf{A}_{12} \\
         \mathbf{A}_{12}^\top & \mathbf{0}
         \end{pmatrix} 
\end{equation*}  
We now treat $\mathbf{A}^{(1)}$ as a small perturbation and apply Rayleigh-Schrödinger perturbation theory. To zeroth order we have $\lambda_2^{(0)} = 0$ and the corresponding eigenvector $\vec v_2^{(0)}$ coincides with the trial vector defined by Eq.~\eqref{eq:trial-vector}.
To linear order, the Fiedler value is then given by
\begin{align*}
   \lambda_2 = \lambda_2^{(0)} + 
       \vec v_2^{(0) \top} \mathbf{L}^{(1)} \vec v_2^{(0)} 
       + \mathcal{O}\left( \mathbf{A}_{12}^2  \right) \\
    = \frac{N_1+N_2}{N_1 \, N_2}  \sum_{e \in S} w_e 
     + \mathcal{O}\left( \mathbf{A}_{12}^2  \right) ,
\end{align*}
where $S = \{  (m,n) \in E | m \in V_1, n \in V_2 \}$ is the cut-set corresponding to the cut $V = V_1 + V_2$.
  
\subsubsection*{Dual Graph}

We can now generalize the above treatment to the dual graph. As discussed above, we base the analysis of communities and connectivity on the dual graph that excludes the exterior face. The Laplacian of the dual graph is defined by \eqref{eq:Laplacian-Dual-DA-i} or its weighted counterpart. In particular, we have 
\begin{equation*}
    \mathbf{L}_i^* = \mathbf{D}^* - \mathbf{A}^*.
\end{equation*}
where the entries of the adjancency matrix are given by
\begin{equation}
     A^*_{f,g} = \sum_{\mbox{edges $e$ shared by faces $f$ and $g$}} 
     \frac{1}{w_e} ,
     \label{eq:adjacency-dual}
\end{equation}
and $D^*_{ff} = \sum_g A^*_{fg}$.

As before, we want to define a decomposition of the graph and use the connectivity at the boundary to derive a bound for the algebraic connectivity. However, we now need the weights \emph{along} the boundary, not \emph{across} the boundary. To make this precise, we first define a cut-path that we will use instead of a cut-set.

\begin{definition}
Let $G = (V,E)$ be a connected plane graph. The outer boundary $\mathcal{B}$ is defined as the cyclic path consisting of the edges adjacent to the exterior face and the connecting vertices.
\end{definition}

\begin{definition}
Let $G = (V,E)$ be a connected plane graph. A cut-path is a path
\begin{equation}
    p = (v_1,e_1,v_2, \ldots, v_{n+1})
\end{equation}
that includes no edges from the outer boundary $\mathcal{B}$ and that satisfies
\begin{enumerate}
\item either the path is cyclic and contains at most one vertex from the boundary $\mathcal{B}$
\item or the path is acyclic and its start and end are on the boundary, $v_1, v_{n+1} \in \mathcal{B}$.
\end{enumerate}
\end{definition}

We will now show that such a cut-path $p$ indeed cuts the dual graph into two parts. 
Then we show that the sum of weights $w_e^{-1}$ along the cut-path provides an upper bound  for the algebraic connectivity of Fiedler value of the dual.

\begin{lemma}
\label{lem:dual-cut}
Let $G = (V,E)$ be connected plane graph and $G^* = (V^*,E^*)$ its dual graph. We consider a cut of the dual $G^*$, i.e.~a decomposition of the vertex set $V^*$ such that
\begin{equation}
   V_1^* \cup V_2^* = V^* \qquad \mbox{and} \qquad V_1^* \cap V_2^* = \emptyset.
\end{equation}
Then the following statements are equivalent
\begin{enumerate}
\item
The two dual vertex sets $V^*_1$ and $V^*_2$ are connected in $G^*$
\item
The cut set 
\begin{equation}
    S^* = \{  e^* = (v^*,u^*) \in E^* | v^* \in V_1^*, u^* \in V_2^* \}
\end{equation}
is minimal.
\item
The dual of the cut set
\begin{align}
    S &= \{  e \in E | e \mbox{ is dual to an } e^* \in S^*   \}  \nonumber \\
       &= \{  e \in E | e \mbox{ is adjacent to a face } f \in V_1^* 
           \mbox{ and a face } g \in V_2^*  \} .
\end{align}
is the edge set of a cut-path.
\end{enumerate}
\end{lemma}

\begin{proof}
The equivalence of (1) and (2), i.e.~the correspondence of connectedness and minimality of the cut set, is a standard result in graph theory (see, e.g., \cite[Lemma 1.9.4.]{Dies10}).

The equivalence of (1,2) and (3) follows from the geometry of the plane embedding: A vertex $v\in V$ corresponds to a point in the plane $\mathbb{R}^2$, an edge to an arc and a face to a closed polygon. An exact proof for dual graphs that include the exterior face can be found in many text books of graph theory, see e.g.~\cite[Proposition 4.6.1.]{Dies10}. Here we briefly sketch how to extend this result to dual graphs that exclude the outer face. 

$(1) \Rightarrow (3)$  
Consider the embedding of the graph in the plane $\mathbb{R}^2$. Then all faces of the graph define a connected area $A \subset \mathbb{R}^2$, whose boundary $\partial A$ corresponds to the edges in the boundary set $\mathcal{B}$. Similar, the faces in $V_{1,2}^*$ correspond to areas $A_{1,2} \subset \mathbb{R}$ such that $A_1 \cup A_2 = A$. As $V_{1,2}^*$ are connected by assumption, so are the areas $A_{1,2}$. Examples of the possible geometries are shown in Supplementary Figure \ref{fig:dual-cut}.
 
The boundary between the two areas $\partial A_{1,2} = A_1 \cap A_2$ is a connected polygonal chain. It contains at most two points in the boundary $\partial A$ of the area $A$, at most one if it is closed, cf.~Supplementary Figure \ref{fig:dual-cut}. Now the boundary coincides with the set
\begin{align*}
    S &=\{  e \in E | e \mbox{ is adjacent to a face } f \in V_1^* 
           \mbox{ and a face } g \in V_2^*  \} .
\end{align*}
which thus inherits the properties of $\partial A_{1,2}$: (i) As $\partial A_{1,2}$ is a connected polygonal chain, $S$ is the edge set of a path. (ii) As $\partial A_{1,2}$ contains at most two points in  $\partial A$, $S$ contains at most two vertices from $\mathcal{B}$, at most one if the path is cyclic.

$(3) \Rightarrow (1)$  
Let $S$ be the edge set of a cut-path $p$. We denote this cut-path as $p = (v_1,e_1,v_2, \ldots, v_{n+1})$ and distinguish two cases.
(a) If the cut-path $p$ is acyclic, its start $v_1$ and end $v_{n+1}$ are distinct and both lie on the boundary $\mathcal{B}$. We can thus write the boundary as
\begin{equation*}
     \mathcal{B} = (u_1 = v_1,\ell_1,\ldots,u_{a}=v_{n+1},\ell_a, u_{a+1},\ldots, u_{b+1}=u_1=v_1).
\end{equation*}
and define two cyclic paths
\begin{align*}
    \tilde p_1 &= (v_1,e_1, \ldots, v_{n+1} = u_a, \ell_a, u_{a+1},\ldots u_{b+1}=v_1), \\
    \tilde p_2 &= (v_1,e_1, \ldots, v_{n+1} = u_a, \ell_{a-1}, u_{a-1}, \ldots, u_1=v_1).
\end{align*}
In the plane embedding, $\tilde p_1$ and $\tilde p_2$ correspond to two closed polygon chains enclosing two polygons $A_{1,2}$. As the $\tilde p_{1,2}$ are cyclic paths, the polygons $A_{1,2}$ are connected. We can again identify the areas $A_{1,2}$ with the dual vertex sets $V^*_{1,2}$ and conclude the following statements.
(i) As $\tilde p_{1,2}$ are cyclic paths, the polygons $A_{1,2}$ are simply connected areas. Hence, the dual vertex sets $V^*_{1,2}$ are connected.
(ii) As $\tilde p_{1,2}$ contain the boundary $\mathcal{B}$, we have $A = A_1 \cup A_2$ and hence $V^* = V_1^* \cup V_2^*$. 
(iii) As the paths $\tilde p_{1,2}$ do not intersect, the two polygons are disjoint, $A_1 \cap A_2 = \emptyset$, and so are the dual vertex sets, $V_1^* \cap V_2^* = \emptyset$.

(b) If the cut-path $p$ is cyclic, its start and end coincide and lie on the boundary $v_1=v_{n+1} \in \mathcal{B}$. 
In the plane embedding, the cut-path $p$ corresponds to a closed polygon chain enclosing a polygon $A_{1}$. As $p$ is a cyclic path, $A_1$ is connected. Notably, $A_1$ touches the boundary in exactly one point corresponding to the node $v_1=v_{n+1}$ according to the definition of a cut-path. The area $A_2 = A \backslash A_1$ must also be connected, otherwise $A_1$ would touch the boundary in more than one point. We can now identify the areas $A_{1,2}$ with the dual vertex sets $V^*_{1,2}$ and conclude that both $V^*_{1}$ and $V^*_{2}$ are connected and satisfy $V_1^* \cap V_2^* = \emptyset$ and $V^* = V_1^* \cup V_2^*$ as desired. 
\end{proof}

\begin{corollary}
Let $G = (V,E,W)$ be a weighted plane graph and
\begin{equation}
    p = (v_1,e_1,v_2, \ldots, v_{n+1})
\end{equation}
be a cut-path. Then the algebraic connectivity of the dual graph $G^*$ satisfies
\begin{equation}
   \lambda_2^* \le \frac{N^*_1+N^*_2}{N^*_1 \, N^*_2}  \sum_{e \in E(p)} w_e^{-1},
\end{equation}
where $E(p) = \{ e_1, \ldots, e_n \}$ is the set of edges in the cut-path and $N_{1,2}^* = |V_{1,2}^*|$
\end{corollary}

\begin{proof}
In full analogy to the analysis of the primal graph (cf.~\eqref{eq:fiedler-primal-bound}), we find the inequality
 \begin{equation*}
   \lambda_2^* \le \frac{N^*_1+N^*_2}{N^*_1 \, N^*_2} 
       \sum_{c \in V^*_1, d \in V^*_2} A^*_{c,d} ,
\end{equation*}
where $V^*_1$ and $V^*_2 = V^* \backslash V^*_2$is a decomposition of the vertex set of the dual and $N^*_{1,2} = |V^*_{1,2}|$. The entries of the dual adjacency matrix are given by Eq.~\eqref{eq:adjacency-dual} such that we have
 \begin{align*}
   \lambda_2^* &\le \frac{N^*_1+N^*_2}{N^*_1 \, N^*_2} 
       \sum_{c \in V^*_1, d \in V^*_2}  \quad 
       \sum_{\mbox{edges $e$ shared by faces $c$ and $d$}} 
        \; w_e^{-1} \\
        &= \frac{N^*_1+N^*_2}{N^*_1 \, N^*_2} 
       \sum_{e \in S}  w_e^{-1} \\
\end{align*}
where the last equality follows from Lemma \ref{lem:dual-cut}.
\end{proof}

\clearpage
          
\refstepcounter{notecounter}
\subsection*{Supplementary Note \thenotecounter:~Linear flow networks}
\label{ap:linear_flow}

Linear flow networks describe a generic model for many types of supply networks including AC power grids \cite{Wood14,Purc05,Hert06}, DC electric circuits~\cite{bollobas1998,Dorfler2018,Kirchhoff1847}, hydraulic networks~\cite{Hwan96,diaz2016}, and vascular networks of plants~\cite{Kati10}. Assume that the underlying network is given by a graph $G=(E,V)$ with $N=|V|$ nodes and $M=|E|$ edges. Then we assign a potential $\theta_n\in\mathbb{R}$ to each node $n\in V$. This potential has the following interpretation: In AC power grids in the linear power flow approximation, it denotes the voltage phase angle whereas in DC electric circuits, it refers to the voltage level at node $n$. Finally, in hydraulic or vascular networks, $\theta_n$ denotes the pressure at node $n$. Then the flow $F_\ell$ along a link $\ell = (m,n)$ is given by
\begin{equation}
    F_\ell = w_{\ell} (\theta_m - \theta_n),
    \label{eq:flow}
\end{equation}
where $w_{\ell}$ is the weight of the link $\ell$. In AC power grids in the linear power flow approximation, $w_\ell$ is proportional to the link's susceptance, in resistor networks it is given by the link's conductance and in a hydraulic or vascular network it depends on the geometry of the pipe or vein. Note that all these networks are undirected, i.e.~flow is possible in both directions alike, such that $w_{m,n} = w_{n,m}$. Since the flow \eqref{eq:flow} depends  only on the potential drop, the potentials $\theta_n$ are only determined up to a constant phase shift applied to all nodes. This problem may be solved by assigning a reference potential to a preselected slack node $n$, e.g. setting $\theta_n:= 0$. 

Now assume that there is an inflow $P_m$ at every node $m \in V$. The value of $P_m$ is positive if a current, power, or fluid is injected to the node and negative if it is withdrawn from the node. Here and in the following, we assume that the in- and outflows sum to zero, $\sum_i^NP_i = 0$, and call them 'balanced'. The flows $F_\ell$ along the edges $\ell\in E$ may then be determined by using the continuity equation, also referred to as Kirchhoff's current law (KCL)
\begin{equation}
     \sum_{\ell \in E} \tilde{I}_{\ell,n} F_\ell = P_n, \qquad
    \forall \, n \in V.
    \label{eq:KCL}
\end{equation}
Here, $\tilde{I}_{\ell,n}$ are the entries of the graph's unweighted edge-node incidence matrix \eqref{eq:incidence-unweighted}.  The potentials $\theta_n$ which fulfil the continuity equation~(\ref{eq:KCL}) and the equation for the flows~(\ref{eq:flow}) automatically satisfy Kirchhoff's voltage law which states that the potential drop around any closed loop $\mathcal{C}$ needs to vanish
\begin{equation}
    \sum_{(n,m)\in\mathcal{C}} \theta_n-\theta_m = 0.\label{eq:KVL}
\end{equation}
Defining a vector of flows $\vec{F}=(F_1,..,F_M)^\top \in\mathbb{R}^M$, a vector of potentials $\vec{\theta}=(\theta_1,...,\theta_N)^\top\in\mathbb{R}^N$ and a vector of inflows $\vec{P}=(P_1,...,P_N)^\top\in\mathbb{R}^N$, we can write these relationships in a more compact form. Kirchhoff's current law~(\ref{eq:KCL}) in vectorized form reads as
\begin{equation*}
    \tilde{\mathbf{I}}^\top\vec{F}=\vec{P}.
\end{equation*}
In addition to that, we can write the relationship between flows and potentials in vectorized form,
\begin{equation}
    \vec{F} = \mathbf{W}\tilde{\mathbf{I}}\vec{\theta}.\label{eq:flow_potential}
\end{equation}
Here, $\mathbf{W}=\operatorname{diag}(w_1,...,w_L)\in\mathbb{R}^{M\times M}$ is a diagonal matrix summarizing the edge weights, cf.~Eq.~\ref{eq:weight-matrix}. Now we can put together the last two equations to arrive at a discrete Poisson equation for the nodal potentials
\begin{equation}
    \mathbf{L}\vec{\theta}= \vec{P}.\label{eq:poisson}
\end{equation}
Here, $\mathbf{L} = \tilde{\mathbf{I}}^\top \mathbf{W}\tilde{\mathbf{I}}$ is the Laplacian matrix. 

\subsubsection*{Quantifying the impact of perturbations - primal graph}

We now use the notation developed in the last section to study the effect of perturbations, in particular links failures, and derive the expression for the sensitivity factor $\eta_{i,k,\ell}$ given in the main text. Assume that the inflow at one node $i$ is increased by an amount $\Delta P$ and decreased at another node $k$ by the same amount. The inflow vector changes as $\vec{P} \rightarrow \vec{P}^\prime = \vec{P} +\Delta\vec{\theta}$ with
\begin{equation*}
    \Delta\vec{P} = \Delta P  (\vec e_i-\vec e_k),
\end{equation*}
where $\vec e_i$ denotes the $i$th standard unit vector. As a consequence, the nodal potentials changes as $\vec{\theta} \rightarrow \vec{\theta}^\prime = \vec{\theta} +\Delta\vec{ \theta}$. Both the new and the old potentials have to fulfil the Poisson equation~(\ref{eq:poisson}), Subtracting the two equations, we arrive at an explicit equation for the change in nodal potentials
\begin{equation*}
    \mathbf{L}\Delta\vec{\theta} = \Delta P(\vec e_i-\vec e_k).
\end{equation*}
Note that the Laplacian matrix of a connected graph always has one zero eigenvalue \cite{newman2010}. Thus, we cannot simply invert this equation to calculate the change in the nodal potentials $\Delta\vec{\theta}$. This problem is typically solved by making use of the Moore-Penrose pseudoinverse of the Laplacian matrix $\mathbf{L}^\dagger$ which has properties similar to the matrix inverse~\cite{Moore1920}.

Now we can use Kirchhoff's current law~(\ref{eq:KCL}) to calculate the resulting changes in the flows
\begin{equation*}
    \Delta\vec{ F} = \mathbf{W}\tilde{\mathbf{I}}\mathbf{L}^\dagger(\vec e_i-\vec e_k)\Delta P. 
\end{equation*}
Finally, we can read off the change of the flow $\Delta F_\ell$ on a line $\ell$. We define an indicator vector $\vec{l}_\ell$ as the $\ell$th standard unit vector in $\vec{R}^M$ and obtain
\begin{equation}
    \Delta F_\ell =\vec{l}_\ell^\top\Delta\vec{F} = w_\ell\vec{l}_\ell^\top \tilde{\mathbf{I}}\mathbf{L}^\dagger(\vec e_i-\vec e_k)\Delta P = \sqrt{w_\ell} \vec{l}_\ell^\top\mathbf{I}\mathbf{L}^\dagger(\vec e_i-\vec e_k)\Delta P ,
\end{equation}    
where we replaced the unweighted incidence matrix $\tilde{\mathbf{I}}$ by the weighted matrix $\mathbf{I}$. Then the sensitivity factor reads
\begin{equation}
    \eta_{i,k,\ell} = \frac{\Delta F_\ell}{\Delta P} = \sqrt{w_\ell} \vec{l}_\ell^\top\mathbf{I}\mathbf{L}^\dagger(\vec e_i-\vec e_k).
\end{equation}
Note that the sensitivity factor is purely topological, i.e. it may be calculated only based on the network topology and weights and does \emph{not} depend on the in- and outflows. If the two nodes $i$ and $k$ are the terminal nodes of an edge $p=(i,k)$, we can identify 
\begin{equation*}
    \vec e_i-\vec e_k = \tilde{\mathbf{I}}^\top\vec{l}_p,
\end{equation*}
and the sensitivity factor reads
\begin{align*}
    \eta_{p,\ell} &= w_\ell \vec{l}_\ell^\top\tilde{\mathbf{I}}\mathbf{L}^\dagger\tilde{\mathbf{I}}^\top\vec{l}_p
    = \sqrt{w_\ell/w_p} \vec{l}_\ell^\top {\mathbf{I}}\mathbf{L}^\dagger{\mathbf{I}}^\top\vec{l}_p.
\end{align*}

\subsubsection*{Quantifying the impact of perturbations - dual graph}
\label{sec:dual}

The discussion of perturbations in the in- and outflows developed in the last section was recently extended to the dual graph~\cite{17lodf,ronellenfitsch_dual_2017}. We will briefly cover this here and derive the dual expression for the sensitivity factor $\eta$.
As before, we assume that the inflow at one node $i$ is increased by an amount $\Delta P$ and decreased at another node $k$ by the same amount such $\Delta\vec{ P} = \Delta P  (\vec e_i-\vec e_k)$. By subtracting Kirchhoff's current law~(\ref{eq:KCL}) for the situation before and after the change in inflow, we arrive at
\begin{equation*}
    \tilde{\mathbf{I}}^\top \Delta\vec{F} = \Delta P  (\vec e_i-\vec e_k),
\end{equation*}
which leaves us with an underdetermined system of equations: We have $M$ unknown flow changes $\Delta F_\ell$, but only $N-1$ equations as one is redundant. The general solution can thus be written as 
\begin{equation}
   \Delta\vec{F} = \Delta\vec{ F}_{\text{part}} +  \Delta\vec{ F}_{\text{hom}} \, ,
    \label{eq:dual_flow_changes}
\end{equation}
where $\Delta\vec{F}_{\text{part}}$ is one particular solution of the equation and  $\Delta\vec{ F}_{\text{hom}}$ is an arbitrary solution of the homogeneous equation, i.e.~an arbitrary solution of $\tilde{\mathbf{I}}^\top \Delta\vec{F}_{\text{hom}}= \vec{0}$. We can now use graph duality to parametrize the homogeneous solutions and to single out the physically correct solution.

The kernel of the node-edge incidence matrix is spanned by the cycle-edge incidence matrix
\begin{equation*}
    \tilde{\mathbf{I}}^\top\tilde{\mathbf{C}} = 0,
\end{equation*}
where we here use the unweighted matrix \eqref{eq:cycle-in-unweighted}. Hence, we can write the homogeneous solutions as 
\begin{equation*}
     \Delta\vec{F}_{\text{hom}} = \tilde{\mathbf{C}}\vec{f},
\end{equation*}
where $\vec{f}\in\mathbb{R}^{M-N+1}$ is a vector of cycle flows -- one for each cycle.
The physically correct values of the cycle flows are then determined by Kirchhoff's voltage law (Eq.~\ref{eq:KVL}) which may be written compactly as
\begin{equation}
     \tilde{\mathbf{C}}^{\top}\mathbf{W}^{-1}\Delta\vec{ F} = 0,
     \label{eq:dual_flow_changes2}
\end{equation}
where we used that $\Delta \theta_n - \Delta \theta_m = \Delta F_{nm} / w_{nm}$. Inserting equations~\eqref{eq:dual_flow_changes} and \eqref{eq:dual_flow_changes2}, we then obtain
\begin{align}
       \tilde{\mathbf{C}}^{\top}\mathbf{W}^{-1}\tilde{\mathbf{C}}\vec{f} &= -\tilde{\mathbf{C}}^{\top}\mathbf{W}^{-1}\Delta\vec{ F}_{\text{part}}, \nonumber\\
           \Leftrightarrow\quad\mathbf{L}^*_g \vec{f}&= -\tilde{\mathbf{C}}^{\top}\mathbf{W}^{-1}\Delta\vec{ F}_{\text{part}},\label{eq:dual_poisson}
\end{align}
with the dual Laplacian $\mathbf{L}^*_g=\tilde{\mathbf{C}}^{\top}\mathbf{W}^{-1}\tilde{\mathbf{C}}$. We thus found a discrete Poisson equation for the cycle flows $\vec f$ with the grounded dual Laplacian $\mathbf{L}^*_g$ in direct correspondence to its primal counterpart~(\ref{eq:poisson}). We can now proceed as for the primal graph: Inverting the discrete Poisson equation for the cycle flows and plugging the result into Eq.~(\ref{eq:dual_flow_changes}), we thus arrive at
\begin{equation*}
    \Delta\vec{ F} = (\mathbf{1}_M -\tilde{\mathbf{C}}(\mathbf{L}^{*}_g)^{-1} \tilde{\mathbf{C}}^{\top}\mathbf{W}^{-1})\Delta\vec{ F}_{\text{part}}.
\end{equation*}
Here, $\mathbf{1}_M$ denotes the identity matrix of dimensions $M\times M$. As a last step, we need to determine a particular solution $\Delta\vec{ F}_{\text{part}}$. This can be accomplished as follows~\cite{ronellenfitsch_dual_2017}: Let $\mathcal{T}^{ik}\in\mathbb{R}^{M}$ be a vector determining an (arbitrary) path between the node $i$ with inflow $\Delta P$ and node $k$ with the same outflow. The entries of $\mathcal{T}^{ik}$ are given as follows 
\begin{equation*}
      \mathcal{T}_{e}^{ik} = \left\{
   \begin{array}{r l }
     1& \; \mbox{if edge $e$ is element of the path $i\rightarrow k$},  \\
     -1 & \; \mbox{if reversed edge $e$ is element of the path $i\rightarrow k$},  \\
      0     & \; \mbox{otherwise}.
   \end{array} \right.
\end{equation*}
Then a particular solution is given by $\Delta\vec{ F}_{\text{part}} = \mathcal{T}^{ik}\Delta P$, and we can plug this into the above equation
\begin{equation}
\begin{aligned}
            \Delta\vec{ F} &= (\mathbf{1}_M -\tilde{\mathbf{C}}(\mathbf{L}^*_g)^\dagger\tilde{\mathbf{C}}^{\top}\mathbf{W}^{-1})\mathcal{T}^{ik}\Delta P,\nonumber\\
            \Rightarrow \Delta F_\ell &=\vec{l}_\ell^\top\Delta\vec{ F} = \left(\vec{l}_\ell^\top\mathcal{T}^{ik}-\vec{l}_\ell^\top\tilde{\mathbf{C}}(\mathbf{L}^*_g)^\dagger\tilde{\mathbf{C}}^{\top}\mathbf{W}^{-1}\mathcal{T}^{ik}\right)\Delta P. \label{eq:Ptdf_dual_injections}
\end{aligned}
\end{equation}
If nodes $i$ and $k$ are the two terminal ends of an edge $e=(i,k)\neq \ell$, the path vector may be identified with the edge's indicator vector, $\mathcal{T}^{ik} = \vec{l}_e$. Then we can simplify the above expression for any edge $\ell \neq e$ to
\begin{equation*}
    \Delta F_{\ell} = -w_p^{-1}\vec{l}_\ell^\top\tilde{\mathbf{C}}(\mathbf{L}^*_g)^\dagger\tilde{\mathbf{C}}^{\top}\vec{l}_e\Delta P,
\end{equation*}
Notably, the first expression in Eq.~\eqref{eq:Ptdf_dual_injections} vanishes since $\vec{l}_\ell^\top\vec{l}_e = \delta_{\ell e}$. Thus, we can identify the sensitivity factor $\eta_{i,k,\ell}$ as
\begin{equation*}
    \eta_{i,k,\ell} = -w_e^{-1}\vec{l}_\ell^\top\mathbf{C}(\mathbf{L}^*_g)^\dagger\mathbf{C}^{\top}\vec{l}_e
    = -\sqrt{\frac{w_\ell}{w_e}}\vec{l}_\ell^\top\mathbf{C}(\mathbf{L}^*_G)^\dagger \mathbf{C}^{\top}\vec{l}_e \, .
\end{equation*}

\subsubsection*{Modelling link failures}

The sensitivity factor $\eta_{i,k,\ell}$ may also be used to describe the failure of links. Assume that a link $e$ fails, thus loosing its ability to carry flow and setting the weight to zero, $w_e=0$. Instead of removing the link from the network, we can find an analogous description based on the sensitivity factor~\cite{Wood14,strake2018}.  

Assume that the link $e=(e_1,e_2)$ carries a flow $F_e$ before the outage. Now we model the removal of the link by an in- and outflow at the two ends of the link: To this end, we assume that the line was disconnected from the network and consider a fictitious flow $\hat{F}_e$ which is the result of a (fictitious) inflow $\Delta P = \hat{F}_e$ at the starting node $e_1$ and an outflow of the same amount at the terminal node $e_2$ of link $e$. On the other hand, we can also calculate the flow $\hat{F}_e$ flowing on line $e$ after the injection by a self-consistency argument: It can be calculated using the sensitivity factor
\begin{equation*}
\begin{aligned}
        \hat{F}_e &= F_e +\eta_{e_1,e_2,e}\Delta P = F_e +\eta_{e_1,e_2,e}\hat{F}_e \\\Rightarrow \Delta P &= \hat{F}_e = F_e(1-\eta_{e_1,e_2,e})^{-1}.
\end{aligned}
\end{equation*}
Now we can calculate the change in the flow on another link $\ell$ due to the failure of link $e$ as 
\begin{equation*}
    \Delta F_{\ell} = \eta_{e_1,e_2,\ell} \Delta P = \frac{\eta_{e_1,e_2,\ell}}{1-\eta_{e_1,e_2,e}}F_e.
\end{equation*}
We can thus calculate the flow change $\Delta F_{\ell}$ on any link $\ell$ due to the failure of another link $e$ based on the initial flow on link $e$ and the sensitivity factor $\eta$. The expression $\eta_{e_1,e_2,\ell}(1-\eta_{e_1,e_2,e})^{-1}$ is known as \textit{Line Outage Distribution Factor} in power system security analysis~\cite{Wood14}. Thus, both changes in the inflows and link failures can be captured by the sensitivity factor $\eta$. 

To quantify the effect of (dual) communities on network robustness in linear flow networks we defined the ratio of flow changes in the main text (cf. Ref.~\cite{Kaiser2019}),
\begin{equation*}
    R(\ell,d) = \frac{\langle|\Delta F_{k} | \rangle_d^{k\in \text{O}}}{\langle
  |\Delta F_{r} |
  \rangle_d^{r\in \text{S}}}.
\end{equation*}
Importantly, the ratio can be used to quantify both, the effect of communities on changes in the inflow patterns and on failure spreading.
For an inflow and outflow of $\Delta P$ at the two terminal ends of the link $e$, $e_1$ and $e_2$, respectively, the ratio is calculated as
\begin{equation*}
      R(e,d) = \frac{\langle|\Delta F_{k} | \rangle_d^{k\in \text{O}}}{\langle
  |\Delta F_{r} |
  \rangle_d^{r\in \text{S}}} =\frac{\langle| \eta_{e_1,e_2,k}\Delta P | \rangle_d^{k\in \text{O}}}{\langle
  |\eta_{e_1,e_2,r}\Delta P  |
  \rangle_d^{r\in \text{S}}} =\frac{\langle| \eta_{e_1,e_2,k} | \rangle_d^{k\in \text{O}}}{\langle
  |\eta_{e_1,e_2,r} |
  \rangle_d^{r\in \text{S}}}.
\end{equation*}
On the other hand, if we instead calculate the ratio for the failure of a link $e$ with initial flow $F_e$, we arrive at
\begin{equation}
      R(e,d) = \frac{\langle|\Delta F_{k} | \rangle_d^{k\in \text{O}}}{\langle
  |\Delta F_{r} |
  \rangle_d^{r\in \text{S}}} =\frac{\langle| \eta_{e_1,e_2,k}(1-\eta_{e_1,e_2,e})^{-1}F_e | \rangle_d^{k\in \text{O}}}{\langle
  |\eta_{e_1,e_2,r}(1-\eta_{e_1,e_2,e})^{-1}F_e  |
  \rangle_d^{r\in \text{S}}} =\frac{\langle| \eta_{e_1,e_2,k} | \rangle_d^{k\in \text{O}}}{\langle
  |\eta_{e_1,e_2,r} |
  \rangle_d^{r\in \text{S}}}.\label{eq:flow_ratio_sensitivity}
\end{equation}
Thus, in both cases, the ratio is determined by the sensitivity factor $\eta$ which is in turn governed by the Pseudo-inverse of the Laplacian matrix $\mathbf{L}^\dagger$ or the Pseudo-inverse of the dual Laplacian $(\mathbf{L}^*_g)^{\dagger}$ when formulating the problem in the dual graph.

\subsubsection*{Connectivity structure determines network response to link failures}
\label{sec:perturb}

In this section, we will demonstrate why a weak connection between two components of a network limits the flow changes in one component when a link in the other one fails. Our analysis follows the approach towards perturbation spreading used in Manik et al.~\cite{Manik2017,16redundancy} that is based on Rayleigh-Schr\"odinger perturbation theory~\cite{Ball98}. 

Consider a connected graph $G=(E,V)$ with $N$ nodes and $M$ edges consisting of two subgraphs $G_1=(E_1,V_1)$ and $G_2=(E_2,V_2)$ with $n_1$ nodes and $n_2=N-n_1$ nodes, respectively, that are mutually weakly connected. Here, a weak connection between the two subgraphs may either be realized through a (relatively) weak number of links in case of an unweighted graph or the overall weight of the connections between the two subgraphs being small~\cite{Dies10}.
We then sort the graph's vertices $V$ in such a way that the first $n_1$ vertices belong to the first subgraph $G_1$ and the other $n_2$ vertices belong to the second one $G_2$. We may regard the weak connections between the to subgraphs as a perturbation to the graph in which the two subgraphs are disconnected. In terms of the graph Laplacian $\mathbf{L}$, we thus write 
\begin{equation}
    \begin{aligned}
    \mathbf{L}=\mathbf{L}_0 + \tilde{\mathbf{L}}=
    \begin{pmatrix} 
    \mathbf{L}_1& \mathbf{0}_{n_1\times n_2}\\ \mathbf{0}_{n_2\times n_1} &\mathbf{L}_2
    \end{pmatrix}
    +
    \begin{pmatrix} 
    \mathbf{D}_{12}& -\mathbf{A}_{12}\\ -\mathbf{A}_{12}^\top &\mathbf{D}_{21}
    \end{pmatrix} .
\end{aligned}\nonumber
\end{equation}
Here, $\mathbf{L}_0$ is the graph Laplacian for the graph when disconnecting the two subgraphs, expressed in terms of their Laplacian matrices $\mathbf{L}_1\in\mathbb{R}^{n_1\times n_1}$ and $\mathbf{L}_2\in\mathbb{R}^{n_2\times n_2}$ and $\tilde{\mathbf{L}}$ is the perturbation matrix with the diagonal degree matrices $\mathbf{D}_{12}$ and $\mathbf{D}_{21}$ denoting the degree for the graph connecting the two subgraphs. $\mathbf{A}_{12}$ is the adjacency matrix for this graph that is assumed to be relatively sparse, thus indicating the weak inter-subgraph connections.

To examine the effect of link failures, we need to study the pseudoinverse of this matrix. Therefore, we denote by $\mathbf{X}=\mathbf{L}^\dagger$ the Moore-Penrose pseudoinverse of the overall graph Laplacian and by $\mathbf{X}_1=\mathbf{L}_1^\dagger$ and $\mathbf{X}_2=\mathbf{L}_2^\dagger$ the Moore-Penrose pseudoinverses of the subgraphs' Laplacian matrices $\mathbf{L}_1$ and $\mathbf{L}_2$, respectively. For the inverse $\mathbf{X}_0$ of the unperturbed Laplacian $\mathbf{L}_0$, we then get 
\begin{equation}
    \begin{aligned}
    \mathbf{X}_0=\begin{pmatrix} \mathbf{X}_1& \mathbf{0}_{n_1\times n_2}\\ \mathbf{0}_{n_2\times n_1} &\mathbf{X}_2\end{pmatrix}.
\end{aligned}\nonumber
\end{equation}
In order to calculate the matrix inverse $\mathbf{X}$, we can expand this matrix using the Neumann series (see e.g. Ref.~\cite{ben-israel_generalized_2003}). We note that some issues may arise in this perturbative treatment and the series expansion, as the Moore-Penrose pseudo-inverse is not continuous in general.
However, these this does not apply here, since we can easily circumvent the formal problems by fixing the phase of an arbitrarily chosen slack node $n$ as $\theta_n \equiv 0$ and removing the $n$-th row from the linear set of equations. The resulting set of equations is of full rank and admits a unique solution. The corresponding matrices are again grounded Laplacians~\cite{miekkala_graph_1993} which are invertible.
The matrix inverse is continuous and admits a series expansion as in Equation~\eqref{eq:expansion}.
For the sake of notational simplicity and coherence, we will however stick to the ordinary Laplacians. The matrix inverse $\mathbf{X}$ then reads
\begin{equation}
    \begin{aligned}
    \mathbf{X}&=(\mathbf{L}_0 + \tilde{\mathbf{L}})^{\dagger} =\left[\mathbf{L}_0\left(\mathbf{1} + \mathbf{X}_0\tilde{\mathbf{L}}\right)\right] ^{\dagger}=\left[\left(\mathbf{1} + \tilde{\mathbf{L}}\mathbf{X}_0\right)\mathbf{L}_0\right] ^{\dagger}\nonumber\\
    &=\mathbf{X}_0\left[\left(\mathbf{1} - (-\mathbf{X}_0\tilde{\mathbf{L}})\right)\right] ^{\dagger}= \mathbf{X}_0\sum_{k=1}^\infty (-1)^k(\tilde{\mathbf{L}}\mathbf{X}_0)^k,\label{eq:expansion}
\end{aligned}
\end{equation}
where we inserted the Neumann series in the last step. We can thus approximate the matrix inverse of the graph Laplacian as 
\begin{equation}
    \begin{aligned}
    \mathbf{X} = \mathbf{X}_0 - \mathbf{X}_0\tilde{\mathbf{L}}\mathbf{X}_0 + \mathcal{O}(\tilde{\mathbf{L}}^2),
\end{aligned}\nonumber
\end{equation}
where $\mathcal{O}(\tilde{\mathbf{L}}^2)$ denotes terms of at least order two in the perturbation matrix $\tilde{\mathbf{L}}$.

Now we can use these expressions to calculate a first order approximation for the sensitivity factor $\eta_{i,k,\ell}$ in the weakly connected limit. Assume we are monitoring the flow changes on line $\ell$ with indicator vector $\vec{l}_\ell$ and define $\vec{\nu}_\ell=\mathbf{I}\cdot \vec{l}_\ell = \vec{e}_{\ell_1}-\vec e_{\ell_2}$ as a result of a power transfer along line $k=(k_1,k_2)$ with indicator vector $\vec{l}_k$ and  $\vec{\nu}_k= \mathbf{I}\cdot \vec{l}_k$. We can then calculate the sensitivity factor as
\begin{equation}
    \begin{aligned}
    \eta_{k_1,k_2,\ell} &=w_{\ell}\vec \nu_{\ell}^\top\mathbf X \vec \nu_{k}.
\end{aligned}\nonumber
\end{equation}
Now we distinguish two cases. First assume that $\ell$ and $k$ are contained in the same subgraph, say $G_1$. In this case, we can write with slight abuse of notation  $\vec{\nu}_\ell=(\hat{\vec{\nu}}_\ell,\vec{0}_{n_2})^\top$ and $\vec{\nu}_k=(\hat{\vec{\nu}}_k,\vec{0}_{n_2})^\top$, where $\hat{\vec{\nu}}_l \in \mathbb{R}^{n_1}$ denotes the projection of the vector onto the subspace of vertices the first subgraph $G_1$ and the whole vector $\vec{\nu}_\ell$ is still understood as a vector in $\mathbb{R}^N$. In this case, we may write the sensitivity factor as
\begin{equation}
    \begin{aligned}
    \eta_{k_1,k_2,\ell}&=w_{\ell}(\hat{\vec{\nu}}_\ell^\top,\vec{0}_{n_2}^\top)\mathbf X     
    \begin{pmatrix}
\hat{\vec{\nu}}_k\\
\vec{0}_{n_2}
    \end{pmatrix} 
    =w_\ell\hat{\vec{\nu}}_\ell^\top\mathbf{X}_1\hat{\vec{\nu}}_k +\mathcal{O}(\tilde{\mathbf{L}}).
    \label{eq:PTDF_sc}
\end{aligned}
\end{equation}
Now consider the case where $\ell$ and $k$ are contained in different modules of the network. In this case, we may write $\vec{\nu}_\ell=(\hat{\vec{\nu}}_\ell,\vec{0}_{n_2})^\top$ and $\vec{\nu}_k=(\vec{0}_{n_1},\hat{\vec{\nu}}_k)^\top$ and calculate the sensitivity factor as 
\begin{equation}
    \begin{aligned}
    \eta_{k_1,k_2,\ell}&=w_\ell(\hat{\vec{\nu}}_\ell^\top,\vec{0}_{n_2}^\top)\mathbf X
     \begin{pmatrix} 
    \vec{0}_{n_1}\\
    \hat{\vec{\nu}}_k
  \end{pmatrix} \nonumber\\
    &=w_\ell(\hat{\vec{\nu}}_\ell^\top,\vec{0}_{n_2}^\top)\left[\mathbf{X}_0+\mathbf{X}_0\tilde{\mathbf{L}}\mathbf{X}_0\right]\begin{pmatrix} 
    \vec{0}_{n_1}\\
    \hat{\vec{\nu}}_k
  \end{pmatrix} +\mathcal{O}(\tilde{\mathbf{L}}^2)\nonumber\\
    &=w_\ell\left[0+(\hat{\vec{\nu}}_\ell^\top\mathbf{X}_1,\vec{0}_{n_2}^\top)\tilde{\mathbf{L}}
    \begin{pmatrix}
    \vec{0}_{n_1}\\
    \mathbf{X}_2\hat{\vec{\nu}}_k
    \end{pmatrix}\right]+\mathcal{O}(\tilde{\mathbf{L}}^2).
    \label{eq:PTDF_oc}
\end{aligned}
\end{equation}
We thus observe that the leading order contribution of the perturbation matrix to the sensitivity factor is $\tilde{\mathbf{L}}^0$ if both links are contained in the same module and $\tilde{\mathbf{L}}^1$ if they are contained in different modules.

Now we compare this to the scaling of Laplacian eigenvalues with the perturbation matrix. The first eigenvector of $\mathbf{L}_0$ is given by the constant shift $\vec{v}_1=N^{-1/2}\vec{1}_N$ and has eigenalue $\lambda_1=0$. In the weakly connected limit considering only $\mathbf{L}_0$, the Fiedler eigenvalue vanishes as well $\lambda_2^{(0)}=0$, and has an associated eigenvector~\cite{Manik2017}
\begin{equation}
    \begin{aligned}
    \vec{v}_2^{(0)}=N^{-1/2}(\underbrace{\sqrt{n_2/n_1},...,\sqrt{n_2/n_1}}_{n_1 \text{ times}},\underbrace{-\sqrt{n_1/n_2},...,-\sqrt{n_1/n_2}}_{n_2 \text{ times}})^\top.
\end{aligned}\nonumber
\end{equation}
Here, we use the superscript $(0)$ to denote the eigenvector and eigenvalue in the unperturbed case. A first order estimate for the Fiedler value may thus be calculated by using Rayleigh Schrödinger perturbation theory as
\begin{equation}
    \begin{aligned}
    \lambda_2^{(1)}=(\vec{v}_2^{(0)})^\top\tilde{\mathbf{L}}\vec{v}_2^{(0)}+\mathcal{O}(\tilde{\mathbf{L}^2}).
    \label{eq:Fiedler_estimate}
\end{aligned}
\end{equation}
For weakly connected graphs, we thus expect the sensitivity factor $\eta_{\ell_1,\ell_2,k}$ to scale with the connectivity of the graph in the same way as the Fiedler value of the graph if $\ell$ and $k$ lie in different communities due to the fact that $\lambda_2^{(1)}\varpropto\tilde{\mathbf{L}}$ and expect it to be to leading order independent of the Fiedler value if $\ell$ and $k$ are in the same community.

\subsubsection*{Flow ratio scales with increasing connectivity between weakly connected modules}

In the main part of the manuscript, we consider the ratio $R(\ell,d)$ between the flow changes $\Delta F_{i\rightarrow j}$ within the other community $(i,j) \in \text{O}$ and the flow changes in the same community $(i,j) \in \text{S}$ after a link failure. It is well established that the flow changes typically decay with distance to the failing link~\cite{strake2018,Kett15,Labavic2014,Jung2015}. To be able to neglect this effect on the flow changes, we take the average absolute flow changes at a fixed distance $d$ denoted by $\langle|\Delta F_{i\rightarrow j}|\rangle_{d}^{(i,j)\in \text{S}}$ for the same community and $\langle|\Delta F_{i\rightarrow j}|\rangle_d^{(i,j)\in \text{O}}$ for the other community where we average over all edges $\ell=(i,j)$ within the respective community that are located at an unweighted edge distance of $d$ to the failing link. With this formalism at hand, we can now estimate the scaling of this ratio with the strength of the perturbation. Suppose that link $k=(r,s)$ is failing. Then the flow ratio reads (cf. Eq.~\eqref{eq:flow_ratio_sensitivity})
\begin{equation*}
    \begin{aligned}
R(\ell,d)=\frac{\langle| \eta_{\ell_1,\ell_2,k} | \rangle_d^{k\in \text{O}}}{\langle
  |\eta_{\ell_1,\ell_2,r} |
  \rangle_d^{r\in \text{S}}}.
\end{aligned}
\end{equation*}
 Now we can make use of our results on the scaling of the sensitivity factors with increasing connectivity between the subnetworks. From Eqs.~\eqref{eq:PTDF_sc} and,~\eqref{eq:PTDF_oc} we see that this ratio scales with the first order of the perturbation
\begin{equation}
    \begin{aligned}
    &R(\ell,d) = \frac{\left\langle\left|w_\ell(\hat{\vec{\nu}}_\ell^\top\mathbf{X}_1,\vec{0}_{n_2}^\top)\tilde{\mathbf{L}}
    \begin{pmatrix}
    \vec{0}_{n_1}\\
    \mathbf{X}_2\hat{\vec{\nu}}_k
    \end{pmatrix}\right|\right\rangle_d^{k=(i,j)\in \text{O}}+\mathcal{O}(\tilde{\mathbf{L}}^2)}{\left\langle\left|w_\ell\hat{\vec{\nu}}_\ell^\top\mathbf{X}_1\hat{\vec{\nu}}_r \right|\right\rangle_d^{r=(p,q)\in \text{S}}+\mathcal{O}(\tilde{\mathbf{L}})}\\
    \Rightarrow &R\varpropto \tilde{\mathbf{L}} +\mathcal{O}(\tilde{\mathbf{L}}^2)
\end{aligned}\nonumber
\end{equation}
in the weakly connected limit assuming without loss of generality that the failing link is located in the first subgraph $G_1$, $\ell\in E(G_1)$. Since we deduced in Eq.~\eqref{eq:Fiedler_estimate} that the Fiedler value scales with the perturbation in the first order as well, we expect the two quantities to show a similar scaling with the perturbation matrix $\tilde{\mathbf{L}}$. 

Note that the derivation works exactly the same way for the dual graph $G^*$ and dual Fiedler value $\lambda_2^*$: We simply have to replace the sensitivity factor $\eta_{\ell_1,\ell_2,k}$ by its dual representation.

\clearpage

\refstepcounter{notecounter}
\subsection*{Supplementary Note \thenotecounter:~Fluctuating sink model with additive Dirichlet noise}
\label{ap:dirichlet_noise}

We start by a review of the original single-sink model as formulated by Corson~\cite{Corson2010}. Consider a linear flow network on a graph $G$ with $N$ nodes and $M$ edges summarized in the node set $V$ and edge set $E$. Then we choose one vertex to be the source of the network while all the others are sinks. We assume that the outflow at the sinks fluctuates in time, which is modelled by 
uncorrelated, iid Gaussian variables  $P_{i} \sim \mathcal{N}(\mu,\sigma),~i\in \{2,...,N\}$. Sources and sinks have to balance at any point in time such that 
\begin{equation}
    \sum_{i=1}^N P_i = 0.\label{ap:balancing}
\end{equation}
This equation immediately yields the statistical properties of the source as a consequence of the statistics of the sinks. We label the nodes such that the source has the index $i=1$ and obtain
\begin{align*}
    \langle P_1\rangle = -(N-1)\mu,
\end{align*}
where $\langle \cdot \rangle$ denotes the mean over different realizations of the Gaussian variables. 

The network structure described by the edge weights $w_\ell$ is then optimized to minimize the average dissipation
\begin{equation*}
    \langle D \rangle = \sum_{\ell=1}^M \frac{\langle F_\ell^2\rangle}{w_\ell}.
\end{equation*}
Here, $\langle F_\ell^2\rangle$ is the second moment of the flows which are determined by the statistics of the sources and sinks by virtue of Kirchhoff's equations. However, we assume that the budget for constructing or strengthening edges is limited leading to a resource constraint
\begin{equation*}
    \sum_{\ell \in E}w_{\ell}^\gamma \le 1.
\end{equation*}
Here, $\gamma$ is a cost parameter that controls how expensive it is to increase the weight of an edge. The main challenge when performing the minimization is the interdependence between link weights and flows which cannot be varied independently. Solving the constrained optimization problem with using the method of Lagrange multipliers yields two conditions: First, the flows are determined by the edge weights $w_\ell$ and the inflows $P_i$ via Kirchhoff's equations. Second, the optimal weights are related to the flows as
\begin{equation}
    w_{\ell}=\frac{\langle F^2_{\ell}\rangle^{1/(1+\gamma)}}{\left(\sum_{e\in E}\langle F^2_{e}\rangle^{\gamma/(1+\gamma)}\right)^{1/\gamma}}.
    \label{eq:Corson-wl}
\end{equation}
Corson~\cite{Corson2010} and Katifori et al.~\cite{Kati10} independently came up with a procedure to solve this problem in an iterative self-consistent manor starting from a random initial guess for the edge weights. Then one out the following steps repeatedly until the procedure converges:
\begin{enumerate}
\item The moments of the flows are determined via Kirchhoff's equations. 
\item The weights are updated according to Eq.~\eqref{eq:Corson-wl}.
\end{enumerate}
Notably, the system can admit several local minima, and it depends on the initial guess, to which minimum the iteration converges. 

Now we extend this set-up to a network with multiple sources. Assume that we again label the nodes such that the first $N_s$ nodes are sources and the remaining $N-N_s$ nodes are sinks, which are still uncorrelated, iid Gaussian variables. However, when considering more than one source, $N_s>1$, the distribution of the sources is not completely determined by Equation~(\ref{ap:balancing}), but has additional degrees of freedom. We use this degree of freedom to put additive Dirichlet noise $X_i \sim \text{Dir}(\alpha)$ on the sources. In particular, we assume that
\begin{equation*}
    P_{i} = -\frac{1}{N_s}\sum_{i=N_s+1}^N P_i + K  \left(\frac{1}{N_s}-X_i\right),\quad X_i \sim \text{Dir}(\alpha), i = 1,\ldots,N_s,
\end{equation*}
where $K\in \mathbb{R}$ is a scaling parameter. Importantly, the Dirichlet variables sum to unity for all realizations,
\begin{align}
    \sum_{i=1}^{N_s} X_i = 1.
\end{align}
The parameter $\alpha$ determines the shape of the dirichlet distribution and tunes the variability of the sources. For $\alpha=1$, the distribution is flat such that the individual $X_i$ fluctuate strongly. For $\alpha \gg 1$, the distribution strongly peaks around the mean such that fluctuations of the individual $X_i$ are small. 

The moments of the Dirichlet random variables are given by
\begin{align*}
    \langle X_i\rangle &= \frac{1}{N_s} \\
    \langle X_i^2\rangle &= \frac{(N_s-1)}{N_s^2(N_s\alpha+1)}+\frac{1}{N_s^2} \\
    \langle X_i X_j\rangle &= -\frac{1}{N_s^2(N_s\alpha+1)}+\frac{1}{N_s^2}.
\end{align*}
Note that this results in a vanishing sum over the second moments
\begin{equation*}
    \sum_{i=1}^{N_s}\langle X_iX_j\rangle = 0,
\end{equation*}
by virtue of the definition of the Dirichlet parameter and the first equation. Now we can shift to new random variables $Y_i = K\left(\frac{1}{N_s}-X_i\right)$ with zero mean, $\langle Y_i \rangle = 0$, and the following second moments
\begin{equation*}
\begin{aligned}
       \langle Y_i Y_j\rangle &= K^2\left\langle\left(\frac{1}{N_s}-X_i\right)\left(\frac{1}{N_s}-X_j\right)\right\rangle\\ &= K^2\left(-\frac{1}{N_s^2}+\langle X_iX_j\rangle\right) \\
&= -K^2 \frac{1}{N_s^2(N_s\alpha+1)},\quad i\neq j,\\
\langle Y_i^2\rangle &=  K^2\left\langle\left(\frac{1}{N_s}-X_i\right)^2\right\rangle = K^2\frac{(N_s-1)}{N_s^2(N_s\alpha+1)}
\end{aligned}
\end{equation*}
such that the sum over the second moment still vanishes $\sum_{i=1}^{N_s}\langle Y_iY_j\rangle = 0$. 

With these results we can compute the second moments of the inflows $\langle P_i P_j \rangle$. Given the network topology and the edge weights $w_\ell$, we can further compute the moments of the edge flows $\langle F_\ell^2 \rangle$ using Kirchhoff's equations. Then one can proceed as in the original model as outlined above. Starting from an initial guess of the edge weights, we iterate the following two steps until convergence:
\begin{enumerate}
\item The moments of the flows $\langle F_\ell^2 \rangle$ are determined from the moments $\langle P_i P_j \rangle$ via Kirchhoff's equations. 
\item The weights are updated according to Eq.~\eqref{eq:Corson-wl}.
\end{enumerate}

\clearpage

\refstepcounter{notecounter}
\subsection*{Supplementary Note \thenotecounter:~Clustering in primal and dual graphs}

While we used spectral clustering to find communities in both primal and dual graphs to enable an analytical treatment in the main manuscript, there are various different algorithms to determine the community structure in a given network.
A commonly used and efficient algorithm is the \textit{Louvain Community Detection Algorithm}, which is based on finding a community partition by iteratively optimizing the modularity of this partition.
We are using the algorithm implemented in Networkx-2.8.6 \cite{hagberg2008exploring}, which uses a combination of \cite{traag2019louvain} and \cite{blondel2008fast}.

The algorithm consists of two steps that are run iteratively.
Starting from a partition in which every node is in its own community, the change in modularity $\Delta Q$ for removing an isolated node $i$ from their lonely community and moving them to a neighboring community $\mathcal{C}$ is calculated using
\begin{align*}
    \Delta Q &= \frac{k_{i,\text{in}}}{2m} - \xi \; \frac{\mathcal{S}_\text{tot} \cdot k_i}{2m^2}
\end{align*}
where $m$ is the total number of edges in the graph, $k_{i, \text{in}}$ is the sum of the weights of the edges connecting node $i$ to the community $\mathcal{C}$, and $\mathcal{S}_\text{tot}$ is the sum of the weights of all edges connected nodes in $\mathcal{C}$.
A resolution parameter $\xi$ smaller than $1$ favors larger communities, while $\xi$ larger than $1$ favors smaller communities.
The first step continues moving isolated nodes to neighboring communities until no further increase in modularity gain can be achieved.
In the second step, the previously found communities are interpreted as nodes of a new network, and the weights of the new edges that connect the new nodes are given be the sum of the edges that connected the different communities.
Subsequently, the first step can be executed again using this new network to find larger communities. 
These two steps are run iteratively until no increase in modularity that is larger than a small threshold is achieved.

Choosing a resolution parameter of $\xi=0.01$ and running the clustering algorithm for the primal and dual graph of the \textit{Acer platanoides} leaf, we find the community structure of the given graphs and compare them with the ones found by hierarchical spectral clustering.
The results are shown in Fig.~\ref{fig:clustering_primal_dual}.
Note, the set of edges that separate the communities are the edges connecting the different communities in the case of primal graphs, while they are along a path in the case of dual communities.
In both cases the primal (i.e. original) graphs are shown in the background in each panel in Fig.~\ref{fig:clustering_primal_dual}.
The separation along veins is even more pronounced and accurate when using the Louvain method, indicating that the description as dual graphs reveals important structural features of the  topology of leaf venation networks.